%% file: redes_proyectivas.tex
\documentclass[10pt,twocolumn,twoside]{article}
\usepackage[utf8]{inputenc} 
\usepackage{amsfonts} 
\usepackage{amssymb}
\usepackage{amsmath}
\usepackage{graphicx} 
\usepackage{listings} 
\usepackage{amsthm} 
\usepackage{color}

\usepackage[cm]{fullpage}

\usepackage{float}
\usepackage{tikz}
\usepackage{pgfplots}
\usepackage{fp}
\usepackage{todonotes}
\usetikzlibrary{decorations.text,external,calc}
\usetikzlibrary{fixedpointarithmetic}

\ifnum\pdfshellescape=1
	\tikzset{external/export=false}
\else
\fi

\usepackage[bookmarks=true]{hyperref} 

\newcommand\titulo{Projective Networks: Topologies for Large Parallel Computer Systems}
\newcommand\autor{\href{http://www.alumnos.unican.es/ccc66}{Cristóbal Camarero Coterillo}}
\hypersetup{
    unicode=false,          
    pdftoolbar=true,        
    pdfmenubar=true,        
    pdffitwindow=true,      
    pdftitle={\titulo},     
    pdfauthor={\autor},     
    pdfsubject={},			
    pdfnewwindow=true,      
    pdfkeywords={},			
	pdfpagemode=None,		
    colorlinks=false,       
    linkcolor=red,          
    citecolor=green,        
    filecolor=magenta,      
    urlcolor=cyan           
}

\newcommand{\R}{\mathbb{R}}

\newtheorem{theorem}{Theorem}[section]

\newtheorem{lemma}[theorem]{Lemma}
\newtheorem{definition}[theorem]{Definition} 
\newtheorem{remark}[theorem]{Remark} 
\newtheorem{notation}[theorem]{Notation}
\newtheorem{example}[theorem]{Example}

\newcommand{\firstuse}[1]{\textsl{#1}}
\newcommand{\latin}[1]{\textit{#1}}


\begin{document}
\title{\titulo}
\author{Cristóbal Camarero, Carmen Martínez, Enrique Vallejo, and Ramón Beivide\thanks{C. Camarero, C. Mart\'{\i}nez, E. Vallejo, and R. Beivide are with the Department of Computer Science and Electronics, Universidad de Cantabria, UNICAN, Spain.
email: cristobal.camarero@unican.es; carmen.martinez@unican.es; enrique.vallejo@unican.es; ramon.beivide@unican.es}
}
\maketitle%

\begin{abstract}
The interconnection network comprises a significant portion of the cost of large parallel computers, both in economic terms and power consumption.
Several previous proposals exploit large-radix routers to build scalable low-distance topologies with the aim of minimizing these costs. However, they fail to consider potential unbalance in the network utilization, which in some cases results in suboptimal designs.
Based on an appropriate cost model, this paper advocates the use of networks based on incidence graphs of projective planes, broadly denoted as Projective Networks. Projective Networks rely on highly symmetric generalized Moore graphs and encompass several proposed direct (PN and demi-PN) and indirect (OFT) topologies under a common mathematical framework.
Compared to other proposals with average distance between 2 and 3 hops, these networks provide very high scalability while preserving a balanced network utilization, resulting in low network costs.
Overall, Projective Networks constitute a competitive alternative for exascale-level interconnection network design.
\end{abstract}


\input{introduction}

\input{model}

\input{projective}

\input{topologies}

%
\input{comparative}

\input{indirect}

\input{conclusions}


\bibliographystyle{IEEEtranS}
\bibliography{main}

\end{document}

%% file: introduction.tex
\section{Introduction}\label{sec:intro}

One current trend in research for the design of Exascale systems is to greatly increase the number of compute nodes. The cost and power of the network of these large systems is significant, which urges to optimize these parameters. Specifically, the problem is how to interconnect a collection of compute nodes using a given router model with as small cost and power consumption as possible. If the interconnection network is modelled by a graph, where nodes represent the routers and edges the links connecting them, the Moore bound can be very useful. The present paper deals with graphs attaining or approaching the generalized Moore bound~\cite{Sampels} while minimizing cost and power consumption.

Graph theory has dealt with very interesting topologies that have not yet been adopted as interconnection networks. One paradigmatic example are Moore graphs~\cite{Miller}. Hoffman and Singleton provided in~\cite{Hoffman} some few examples of regular graphs of degree $\Delta$ and diameter $k$ having the maximum number of vertices; namely for $k=2$ and $\Delta=2, 3, 7$ and for $k=3$ and $\Delta=2$. They denoted such graphs as \textsl{Moore graphs} as they attain the upper bound for their number of nodes, solving for these cases, the ($\Delta$-$k$)-problem posed by E. F. Moore. Such graphs are optimal for interconnection networks as they simultaneously minimize maximum and average transmission delays among nodes.

In these interconnection networks, traffic is frequently uniform; when it is not, it can be randomized (using Valiant routing,~\cite{Valiant_ACM}). Under uniform traffic, maximum throughput depends on the network average distance $\bar k$, rather than the diameter $k$. This
promotes the search of generalized Moore graphs~\cite{Sampels}, which have minimum average distance for a given degree.
This is attained when, from a given node, there are the maximum amount of reachable nodes at any distance lower than the diameter, with the remaining nodes at distance  $k$.

As it will be shown in this paper, Moore and some generalized Moore graphs also minimize cost. If it is assumed that network cost is dominated by the number of employed ports (especially SerDes, as it will shown next), minimizing graph average distance not only maximizes throughput but it can also minimize investment and exploitation expenses. Nevertheless, it is important to highlight that highly symmetric graphs are always preferable as they do not exhibit bottlenecks that can compromise performance under uniform traffic. This paper shows examples of such topologies based on incidence graphs of projective planes and compares them with competitive alternatives. Incidence graphs of finite projective planes~\cite{Brown},~\cite{Erdos} have been used to attain the Moore bound, but not only mathematicians have paid attention to this discrete structures. In fact, Valerio \latin{et al.} already use them to define Orthogonal Fat Trees (OFT)~\cite{Valerio}, which are highly scalable cost optimal indirect networks. Brahme \latin{et al.}~\cite{Brahme} propose other topologies for direct networks for HPC clusters. Al it is shown in this paper, they can also be defined using projective planes, although the authors use perfect difference sets for their definition. In this paper it is shown how incidence graphs of finite projective planes are suitable topologies for both direct and indirect networks for HPC systems.

Recently, three strongly related papers have been published. We summarize next their main achievements and bring to light how the results introduced  in our paper improve them. In~\cite{Rumley2015}, the authors propose a methodology based on minimizing average distance to identify optimal topologies for Exascale systems. Therefore, topologies close to the generalized Moore bound are searched. In this aim, several compositions (Cartesian graph products in general) of known topologies are explored. However, in this analysis neither the symmetry nor the link utilization of the topologies are included and, therefore, the comparison may not reflect actual network performance.
in~\cite{Besta} the Slim Fly (SF) network is proposed. This topology provides very high scalability for diameter 2, approaching the Moore bound. 
However, SF is neither symmetric nor well-balanced. Therefore, the number of compute nodes per router must be adjusted in order to give full bisection bandwidth. Moreover, this lack of symmetry makes SFs more costly than projective networks with the same diameter, which also provide higher scalability. 
Finally, in~\cite{Kathareios} several diameter 2 topologies are studied, namely Stacked Single-Path Tree, Multi-layer Full-Mesh, Slim Fly and Two-Level Orthogonal Fat Tree. The authors present experimental results which conclude that the Slim Fly and the OFT are the best direct and indirect topologies respectively. The present paper proves that topologies with diameter other than 2 such as projective networks are also interesting. Furthermore, a more accessible construction of the OFT and its relation with other topologies is given.

The rest of the paper is organized as follows. Section~\ref{sec:model} describes the cost model
assumed in this paper. An expression based on average distance and link utilization which upper bounds the cost is obtained. As it will be shown, maximizing the number of terminals while maintaining the average distance and link utilization will be the target, which will be related to the generalized Moore bound. In Section~\ref{sec:projective} \textsl{Projective Networks} are introduced, 
defined using incidence graphs of projective planes with the smallest average distance for their size and higher symmetry.
In Section~\ref{sec:topologies} a thorough analysis of how graph theoreticians have solved the generalized Moore bound for diameters 1--6 is done. This allows to present a complete comparative, in terms of our power/cost model, of all these topologies in Section \ref{sec:comparative}, with special emphasis on the diameter 2 case.
In Section~\ref{sec:indirect} the case for indirect networks is considered. The cost model is adapted for indirect networks of diameter 2. As it will be shown, optimal topologies can also be obtained with our methodology to derive projective networks. Finally, in Section~\ref{sec:conclusions} the main achievements of the paper are summarized.

%% file: model.tex
\section{Power and cost optimization}\label{sec:model}

The interconnection network constitutes a significant fraction of the cost of an High Performance Computing (HPC) or datacenter system, both in terms of installation 
and operation, 
with the latter mainly dominated by energy costs. This section introduces a coarse-grain generic cost model based on the network average distance and average link utilization. This cost model will be employed to compare different topologies in next sections.

A network should provide the required bandwidth to its collection of compute nodes with minimal latency, while scaling to the required size. Measures of interest are throughput and average latency under uniform traffic. This uniform traffic pattern not only determines the topological properties of the network, but also appears in multiple workloads (such as data-intensive applications or in many collective primitives) and determines the worst-case performance when using routing randomization~\cite{Valiant_ACM}.

An important figure in the deploying of a network is the number of ports in each router chip, also called router radix. This number is a technological constraint, and current 100~Gbps designs typically only support 32 to 48 ports~\cite{Broadcom2015,MellanoxSB7790-2015,Derradji2015,OmniPath2015}. Different configurations of these  switches, or alternative designs~\cite{Chrysos}, provide more than a hundred ports but at lower speeds, typically 25~Gbps. Larger non-blocking routers are built employing multiple routing chips, at the cost of an increased complexity and at least triple switching latency~\cite{MellanoxCS7500-2015,OmniPathDirector2015}.

Thus, our goal will be to build a network for $T$ computing nodes using routers of radix $R$, able to manage uniform traffic at full-bisection bandwidth and minimizing its cost.
Therefore, the use of the expression \textsl{optimal network} along this document refers to this optimization problem.
Let us consider next in more detail such requirements.

For simplicity, all links are assumed to have the same transmission rate, not only links between routers but also links from computing nodes.
The notation used throughout the paper is presented in Table~\ref{tbl:notation}. $\Delta$ is employed to refer to the degree of a graph $G$; when $G$ is a $\Delta$-regular graph, $2|E(G)|=N\Delta$. Similarly, $\Delta_0$ is generally equal to all routers; in such case the router radix is $R=\Delta+\Delta_0$ and the number of compute nodes $T=N\Delta_0$.

\begin{table}
    \centering
    \resizebox{\linewidth}{!}{%
        \begin{tabular}{|c|l|}
          \hline
          Parameter & Definition \\
          \hline
          $T$ & Number of compute nodes or terminals. \\
          $R$ & Router radix (number of ports). \\
          $G(V,E)$ & Graph whose vertices $V$ represent the routers\\
            &  and its edges $E$ the connection between routers.\\
          $N=|V|$ & Number of routers.\\
          $\Delta$ & Maximum degree of $G$. \\
          $\Delta_0$ & Number of compute nodes attached to every router.\\
          $k$ & diameter of $G$.\\
          $\bar k$ & Average distance of $G$.\\
          $a$ & Load accepted by each router in saturation.\\
          $u$ & Average utilization of links.\\
          \hline
        \end{tabular}
    }\vskip 1em minus .5em
    \caption{Notation used in the paper.}
	\label{tbl:notation}
\end{table}

\subsection{Network Dimensioning and Cost Model}\label{sect:cost-models}

In this subsection a generic cost model for both power and hardware required by the network is introduced. This cost depends not only on the average distance of the topology, but also on the average utilization of the network links. Previous works such as \cite{Besta, Rumley2015} do not consider network utilization in their calculations, what leads to suboptimal results.

First, the number of compute nodes $\Delta_0$ which can be serviced per router is estimated. In this aim, \textsl{ideal routers} with minimal routing and a uniform traffic pattern will be assumed.
As the load $a$ increases, the saturation point is reached when some network link becomes in use all the cycles.
When this happens, the network links will have an average utilization $u\in(0,1]$. If $u=1$ then $G$ is said \firstuse{well-balanced}. Being $G$ \firstuse{edge-transitive} is a sufficient but not necessary condition to be well-balanced~\cite{Camarero_thesis}.

If the load injected per cycle per router at saturation is $a$,\footnote{All routers are assumed to inject approximately the same load at saturation.}  then the average utilization $u$ is
	$$u=\frac{load}{\#links}=\frac{aN\bar k}{2|E(G)|}=\frac{a\bar k}{\Delta}.$$
The load in terms of the utilization is $a=\Delta\dfrac{u}{\bar k}$.
Therefore, the number of compute nodes per router $\Delta_0$ which can be serviced without reaching the saturation point is:
\begin{equation}\label{eq:nodes}
	\Delta_0\leq \Delta\frac{u}{\bar k}.
\end{equation}
Ideally, the equality should hold. If Equation (\ref{eq:nodes}) does not hold, the network is said to be \firstuse{oversubscribed}, and does not provide full bisection bandwidth under uniform traffic. Conversely, for $\Delta_0$ lower than the equality value, the network is oversized for the number of compute nodes connected.

Now, a generic estimation for the network cost per computing node $C_{node}$ is considered, which is also particularized to economic or power terms ($C_{node-\$}$ and $C_{node-W}$ in \$ and Watts, respectively).
A generic average cost $c_i$ per injection port, $c_t$ per transit port, and $c_r$ per router are assumed. The resultant cost per compute node is
$$C_{node} = \frac{N}{T}\cdot(c_i \Delta_0+c_t \Delta+c_r ) = \frac{c_i N\Delta_0+c_t N\Delta+c_r N}{T}.$$

Considering the equality value in Equation (\ref{eq:nodes}), $T=N\Delta_0$ and $R=\Delta+\Delta_0$, it results:
\begin{equation}\label{eq:cost}
C_{node} = c_i + c_t\frac{\bar k}{u} + c_r \frac{1+\bar k/u}{R}.
\end{equation}

For the installation cost $C_{node-\$}$, router and transit links comprise the largest amounts. The router cost is roughly proportional to the number of ports, so it contributes a large amount to $c_i, c_t$ and a small amount to $c_r$~\cite{Besta}. Regarding links, as network speed increases optics are expected to displace copper for even shorter distances, including both intra-rack and on-board communications~\cite{Doany2014}. 
When all network links are active optical cables their cost is largely independent of their length, since it is dominated by the optical transceivers in the ends. This leads to $c_i=c_t >> c_r$, with  $c_i=c_t$ approximately constant.
Therefore, the largest component of the installation cost in Equation (\ref{eq:cost}) will be determined by the router ports, $C_{node-\$} \approx c_t(1+\frac{\bar k}{u})$. A more detailed analysis considering different types of cables is presented in Section~\ref{sec:comparative}.

For the energy cost $C_{node-W}$, the most significant part are the router SerDes (which imply large $c_i, c_t$ and small $c_r$); for example, the router design in~\cite{Chrysos} dedicates 87\% of its power to SerDes. Again, this leads to the same result as for the installation cost, concluding that the best cost is obtained using
topologies that minimize $\frac{\bar k}{u}$.

\subsection{Moore Bounds}

In this subsection limits of the network size and its cost will be studied. This will be done by considering the limits of the Moore bound for the relation between the diameter and network size, and the generalized Moore bound for the relation between the average distance and network size, both for a given degree.

Section~\ref{sect:cost-models} concludes that cost depends linearly on $(1+\bar k/u)$.
This expression is minimized in the complete graph $K_N$, which is symmetric---hence $u=1$---and has minimum average distance $\bar k=1$.
However, the complete graph has $\Delta_0=N$, $R=2N-1$ and $T=N^2=\left(\frac{R+1}{2}\right)^2$. 
With a radix $R=48$ the number of compute nodes would be only $T \approx 576$ nodes. 

The Moore Bound~\cite{Miller} establishes that for a given diameter $k$ the maximum network size is bounded by:
\begin{equation}\label{eq:moore}
	N\leq M(\Delta,k)=\frac{\Delta(\Delta-1)^k-2}{\Delta-2}.
\end{equation}
This bound is obtained by assuming the following distance distribution---the number $W(t)$ of vertices at distance $t$ from any chosen vertex:
	$$
	W(t)=\begin{cases}
	1						&\text{if $t=0$}\\
	\Delta(\Delta-1)^{t-1}	&\text{otherwise.}
	\end{cases}
	$$
Therefore, the average distance of a Moore graph is
		$$\bar k=\frac{\sum_{t=1}^ktW(t)}{N-1}=\frac{\sum_{t=1}^k\Delta(\Delta-1)^{t-1}}{N-1}.$$
Then, it is straightforward that $\lim_{\Delta\rightarrow\infty}\bar k=k$.
There are good families of graphs approaching the Moore bound for low diameter, but they are restricted to very specific values in the number of nodes. Additionally, as derived from Equation (\ref{eq:cost}), the most important factor to minimize cost is the average distance $\bar k$, not the network diameter.

\firstuse{Generalized Moore graphs}~\cite{Sampels} reach the minimum average distance for a given router radix and number of vertices $N$. This is attained when there are the maximum amount of reachable nodes up to distance $k-1$, with the remaining nodes being at distance $k$. That is, with the following distance distribution:
	$$
	W(t)=\begin{cases}
	1						&\text{if $t=0$}\\
	\Delta(\Delta-1)^{t-1}	&\text{if $1\leq t\leq k-1$}\\
	N-M(\Delta,k-1)			&\text{if $t=k$.}
	\end{cases}
	$$
With this generalization, the average distance can be approximated---for large $\Delta$---as
\begin{equation}\label{eq:average}
	\bar k\approx k-\frac{\Delta^{k-1}}{N}.
\end{equation}

The generalized Moore bound determines the minimal average distance $\bar k$ (hence cost, given a well-balanced topology) for a given number of nodes $T$ and router radix $R$. Next, an expression relating these values and the diameter $k$ is obtained.
Following Equation~\eqref{eq:nodes}, the number of compute nodes per router is $\Delta_0=\Delta/\bar k=(R-\Delta_0)/\bar k$. Thus, $R=\Delta_0\bar k+\Delta_0=\Delta_0(1+\bar k)$ and
	$$\Delta_0=\frac{R}{\bar k+1}.$$
The degree is
	$$\Delta=R-\Delta_0=R\left(1-\frac{1}{\bar k+1}\right)=R\frac{\bar k}{\bar k+1}.$$
The number of routers is
	$$N=\frac{T}{\Delta_0}=\frac{T}{R}(\bar k+1).$$
The difference $k-\bar k$ can be approximated (using Equation (\ref{eq:average})) by
	$$k-\bar k\approx
	\frac{\Delta^{k-1}}{N}
	=\frac{\left(R\frac{\bar k}{\bar k+1}\right)^{k-1}}{\frac{T}{R}(\bar k+1)}
	=\frac{R^k}{T} \frac{\bar k^{k-1}}{(\bar k+1)^k}.
	$$
Reordering terms, it is obtained the relation:
\begin{equation}\label{eq:final}
	T \approx \frac{R^k\bar k^{k-1}}{(k-\bar k)(\bar k+1)^k}
\end{equation}

This equation is used later as an upper bound for the number of compute nodes in direct topologies.

%% file: projective.tex
\section{Projective Networks: A Topology Based on Incidence Graphs of Finite Projective Planes}\label{sec:projective}

As argued in previous section, average distance and average link utilization are the target parameters to design optimal cost topologies.
In this section incidence graphs of projective planes are proposed to define network topologies attaining almost optimal values of these parameters. In Subsection~\ref{subsec:incidence} incidence graphs of finite projective planes are defined, which constitute a family of symmetric graphs with diameter 3 and average distance equal to $2.5$ in the limit. In Subsection~\ref{subsec:brown} such graphs are modified in such a way that their diameter and average distance both become 2. However, they are no longer symmetric although their link utilization equals 1 in the limit. These two families of graphs are used to define \textsl{Projective Networks} which, as it will be show in Subsection~\ref{subsec:discu}, result in a competitive alternative to the recently proposed Slim Fly network \cite{Besta}. Thus, in this section the methodology proposed in the paper is validated by a specific example.

\subsection{Incidence Graph of Finite Projective Planes}\label{subsec:incidence}


A family of graphs with an average distance tending to $2.5$ can be obtained as the incidence graph of finite projective planes. Next, an algorithmic description of these graphs is given, although a more geometrical approach is considered in Example \ref{ex:Fano}. Since these graphs are defined in terms of finite projective planes, let us first introduce this concept.

Let $q$ be any power of a prime number. A canonic set of representatives of the finite projective plane over the field with $q$ elements $\mathbb F_q$ is
	$$P_2(\mathbb F_q)=\{(1,x,y),(0,1,x),(0,0,1)\mid x,y\in \mathbb F_q\}.$$

\begin{remark} By a straightforward counting argument, it can be proved that $P_2(\mathbb F_q)$ has $q^2+q+1$ elements. \end{remark}


Two points $X, Y\in P_2(\mathbb F_q)$ are said \firstuse{orthogonal} (written $X\perp Y$) if their scalar product is zero. The space $P_2(\mathbb F_q)$ contains also $q^2+q+1$ lines of exactly $q+1$ points each. 
Every line is represented by its dual point in the projective plane. A line $L$ is incident to a point $P$ if and only if $P$ is orthogonal to the dual point of $L$. This fact allows the following definition.

\begin{definition} Let $q$ be a power of a prime number. Let $G_q = (V, E)$ be the graph with vertex set
$$V=\{ (s,P) \mid s\in\{0,1\},\ P\in P_2(\mathbb F_q) \}$$
and edges set
	$$E=\bigl\{ \{(0,P),(1,L)\} \mid P\perp L,\ P,L\in P_2(\mathbb F_q) \bigr\}.$$
Thus, $G_q$ is said to be the \textsl{incidence graph of the finite projective plane} $P_2(\mathbb F_q)$.
\end{definition}

\begin{remark}
Incidence graphs, also called \textsl{Levi graphs}, can be applied to any incidence structure~\cite{Gross}.
Note that $G_q$ is the Levi graph with a finite projective plane as the incidence structure.
\end{remark}

It is clear that $G_q$ has $2q^2+2q+2$ vertices. Let us consider the following example to better understand this construction.

\begin{example}\label{ex:Fano} Let us consider the graph $G_2$. In Figure~\ref{fig:LeviPG} two different structures are represented. On the left side, a typical graphical representation of $P_2(\mathbb F_2)$, or the Fano plane, is shown. In this representation, both the 7 points and their incident lines of the Fano plane are labeled with their homogeneous coordinates. Note that the point 100 is incident to the line \emph{001} since the scalar product of their coordinates is zero. On the right side of the figure, a graphical representation of the incidence graph of the Fano plane, denoted by $G_2$, is shown. There are two kinds of vertices, which are the points and the lines of the Fano plane. Now, two vertices are adjacent if the corresponding point and line are incident. Therefore, since point 100 is incident to line \emph{001} as we have seen before, in the graph there is an edge making them adjacent vertices. As it can be seen, every vertex has degree 3 and there are minimal paths of lengths 1, 2 or 3.

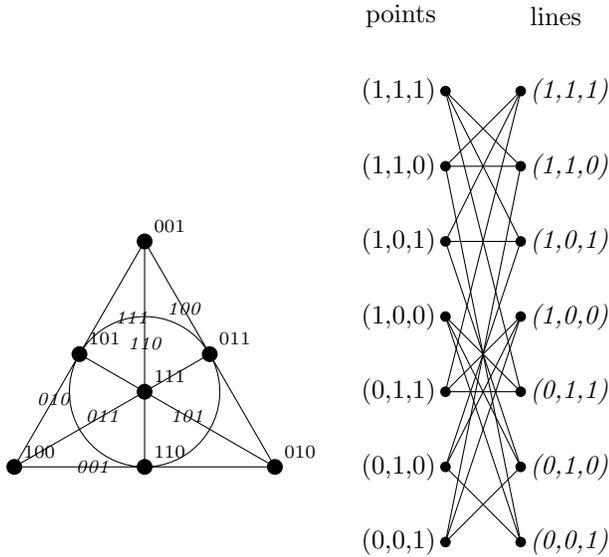
\begin{figure}
	\begin{center}
	\begin{tikzpicture}
	\begin{scope}[
			point/.style={anchor=south west,font=\scriptsize},
			line/.style={font=\scriptsize\em},
			]
		\fill (0,0) circle (3pt) node[point] {111}
			(30:1)  circle (3pt) node[point] {011}
			(90:2)  circle (3pt) node[point] {001}
			(150:1) circle (3pt) node[point] {101}
			(210:2) circle (3pt) node[point] {100}
			(270:1) circle (3pt) node[point] {110}
			(330:2) circle (3pt) node[point] {010};
		\draw (0,0) circle (1) (100:1) node[line] {111}
			(90:2) -- node[line,pos=0.3]{100} (330:2)
			(210:2) -- node[line,pos=0.3] {010} (90:2)
			(210:2) -- node[line,pos=0.3] {001} (330:2)
			(30:1) -- node[line,pos=0.55] {011} (210:2)
			(150:1) -- node[line,pos=0.55] {101} (330:2)
			(270:1) -- node[line,pos=0.55]{110} (90:2);
	\end{scope}
	\begin{scope}[shift={(4,-3)}]
		\foreach \a in {0,1}
		\foreach \b in {0,1}
		\foreach \c in {0,1}
		{
			\pgfmathtruncatemacro\heigh{\a*4+\b*2+\c}
			\ifthenelse{\heigh=0}{}
			{
				\path[fill] (0,\heigh) node[anchor=east] {(\a,\b,\c)} circle (2pt) coordinate (left \a \b \c) ++(1,0) coordinate (right \a \b \c) circle (2pt) node [anchor=west] {\em(\a,\b,\c)};
			}
		}
		\draw (0,8) node[anchor=east] {points} ++(1,0) node [anchor=west] {lines};
		\foreach \la in {0,1}
		\foreach \lb in {0,1}
		\foreach \lc in {0,1}
		\foreach \ra in {0,1}
		\foreach \rb in {0,1}
		\foreach \rc in {0,1}
		{
			\pgfmathtruncatemacro\goodl{\la==1 || \lb==1 || \lc==1}
			\pgfmathtruncatemacro\goodr{\ra==1 || \rb==1 || \rc==1}
			\pgfmathtruncatemacro\dot{mod(\la*\ra + \lb*\rb +\lc*\rc,2)}
			\pgfmathtruncatemacro\good{\goodl && \goodr && \dot==0}
			\ifthenelse{\good=1}
			{
				\draw (left \la \lb \lc) -- (right \ra \rb \rc);
			}
		}
	\end{scope}
	\end{tikzpicture}
	\end{center}
	\caption{Left: the projective plane $P_2(\mathbb F_2)$, also known as the Fano plane. Right: the incidence graph $G_2$, also known as Heawood graph.}
	\label{fig:LeviPG}
\end{figure}

\end{example}

It is known that for any two different points $X,Y\in P_2(\mathbb F_q)$ there is a unique $Z\in P_2(\mathbb F_q)$ such that $X\perp Z$ and $Z\perp Y$.
This implies that the half of the vertices $(0,X)$ of $G_q$ are at distance 2 from $(0,(1,1,1))$ and the other half are at distance at most 3.
$P_2(\mathbb F_q)$ also satisfies that there are $q+1$ orthogonal points to any given one. Thus, in general $G_q$ is a bipartite graph of degree $\Delta=q+1$ with distance distribution
	$$W(t)=\begin{cases}
		1			&\text{if $t=0$}\\
		q+1			&\text{if $t=1$}\\
		q^2+q		&\text{if $t=2$}\\
		q^2			&\text{if $t=3$.}\\
	\end{cases}$$

As a consequence, the average distance of $G_q$ is
	$$\bar k=
	\frac{5q^2+3q+1}{2q^2+2q+1}
	=2.5-\frac{2q+1.5}{2q^2+2q+1}.$$
Thus, the limit of $\overline{k}$ is 2.5 and its diameter $k=3$. Moreover, it can be proved that $G_q$ is symmetric, which gives the optimal average link utilization.

\begin{theorem}\label{thm:symmetry} $G_q$ is symmetric.
\end{theorem}

\begin{proof}
For any invertible matrix $M \in \mathcal{M}_3(\mathbb F_q)$, the application that maps the point $P$ to the point $MP$ is an automorphism of the projective plane $P_2(\mathbb F_q)$, since it maps subspaces to subspaces.
As they preserve the incidence relation, they are also automorphisms of $G_q$.

Now, in order to prove both vertex-transitivity and edge-transitivity, let us prove that for any vertices $(0,P)$, $(1,L)$, $(0,P')$ and $(1,L')$ with $(0,P)$ adjacent to $(1,L)$ and $(0,P')$ adjacent to $(1,L')$ there is a graph automorphism that maps $(0,P)$ into $(0,P')$ and $(0,L)$ into $(0,L')$.
This is equivalent to finding an automorphism $\varphi$ of $P_2(\mathbb F_q)$ that maps the point $P$ into $P'$ and the line $L$ into $L'$.
Let $Q$ be any other point in the line $L$ and $Q'$ any other point in the line $L'$.
By linear algebra there is an invertible matrix $M$ such that $M[P,Q]=[P',Q']$. The induced automorphism is the one desired.
To complete the vertex-transitivity note that mapping $(s,P)$ into $(1-s,P)$ is a graph automorphism.
\end{proof}

An interesting case of $G_q$ graphs is the one in which $q = p^2$ is a square, where $p$ is a power of a prime. In this case, the projective plane $P_2(\mathbb{F}_{p^2})$ can be partitioned into $p^2-p+1$ subplanes $P_2(\mathbb{F}_{p})$~\cite{Hirschfeld}. This implies that $G_{p^2}$ can be partitioned into $p^2-p+1$ graphs isomorphic to $G_p$, which leads to an straightforward layout of the network. Figure~\ref{fig:G_4} shows the partitioning of $G_4$ as an example. In this figure global links are represented with red dashed lines and local links with solid black lines. The local links induce $3 = 2^2-2+1$ subgraphs isomorphic to $G_2$. The label of the vertices refers to the field isomorphism given by $\mathbb F_4\cong \frac{\mathbb F_2[x]}{(x^2+x+1)}.$ Note that the number of global links is almost the square of the local links.

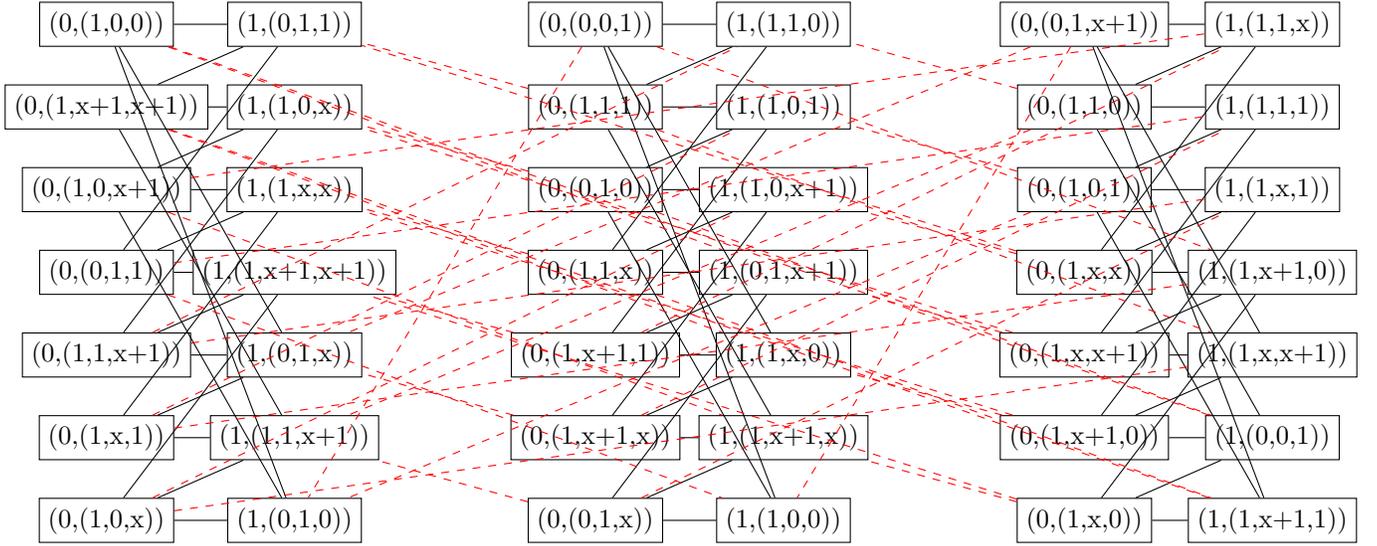
\begin{figure*}
	\begin{center}
	\begin{tikzpicture}[every node/.style={draw}]
	\expandafter\def\csname point0\endcsname {(0,(1,0,0))}
	\expandafter\def\csname point1\endcsname {(0,(0,0,1))}
	\expandafter\def\csname point2\endcsname {(0,(0,1,x+1))}
	\expandafter\def\csname point3\endcsname {(0,(1,x+1,x+1))}
	\expandafter\def\csname point4\endcsname {(0,(1,1,1))}
	\expandafter\def\csname point5\endcsname {(0,(1,1,0))}
	\expandafter\def\csname point6\endcsname {(0,(1,0,x+1))}
	\expandafter\def\csname point7\endcsname {(0,(0,1,0))}
	\expandafter\def\csname point8\endcsname {(0,(1,0,1))}
	\expandafter\def\csname point9\endcsname {(0,(0,1,1))}
	\expandafter\def\csname point10\endcsname{(0,(1,1,x))}
	\expandafter\def\csname point11\endcsname{(0,(1,x,x))}
	\expandafter\def\csname point12\endcsname{(0,(1,1,x+1))}
	\expandafter\def\csname point13\endcsname{(0,(1,x+1,1))}
	\expandafter\def\csname point14\endcsname{(0,(1,x,x+1))}
	\expandafter\def\csname point15\endcsname{(0,(1,x,1))}
	\expandafter\def\csname point16\endcsname{(0,(1,x+1,x))}
	\expandafter\def\csname point17\endcsname{(0,(1,x+1,0))}
	\expandafter\def\csname point18\endcsname{(0,(1,0,x))}
	\expandafter\def\csname point19\endcsname{(0,(0,1,x))}
	\expandafter\def\csname point20\endcsname{(0,(1,x,0))}
	\expandafter\def\csname line0\endcsname  {(1,(0,1,1))}
	\expandafter\def\csname line1\endcsname  {(1,(1,1,0))}
	\expandafter\def\csname line2\endcsname  {(1,(1,1,x))}
	\expandafter\def\csname line3\endcsname  {(1,(1,0,x))}
	\expandafter\def\csname line4\endcsname  {(1,(1,0,1))}
	\expandafter\def\csname line5\endcsname  {(1,(1,1,1))}
	\expandafter\def\csname line6\endcsname  {(1,(1,x,x))}
	\expandafter\def\csname line7\endcsname  {(1,(1,0,x+1))}
	\expandafter\def\csname line8\endcsname  {(1,(1,x,1))}
	\expandafter\def\csname line9\endcsname  {(1,(1,x+1,x+1))}
	\expandafter\def\csname line10\endcsname {(1,(0,1,x+1))}
	\expandafter\def\csname line11\endcsname {(1,(1,x+1,0))}
	\expandafter\def\csname line12\endcsname {(1,(0,1,x))}
	\expandafter\def\csname line13\endcsname {(1,(1,x,0))}
	\expandafter\def\csname line14\endcsname {(1,(1,x,x+1))}
	\expandafter\def\csname line15\endcsname {(1,(1,1,x+1))}
	\expandafter\def\csname line16\endcsname {(1,(1,x+1,x))}
	\expandafter\def\csname line17\endcsname {(1,(0,0,1))}
	\expandafter\def\csname line18\endcsname {(1,(0,1,0))}
	\expandafter\def\csname line19\endcsname {(1,(1,0,0))}
	\expandafter\def\csname line20\endcsname {(1,(1,x+1,1))}
	\foreach \i in {0,...,20}
	{
		\pgfmathsetmacro\subplanepos{int(\i/3)}
		\pgfmathsetmacro\subplane{int(mod(\i,3))}
		\node (point\i) at (6.5*\subplane,-\subplanepos*1.1) {\csname point\i\endcsname};
		\node (line\i) at (6.5*\subplane+2.5,-\subplanepos*1.1) {\csname line\i\endcsname};
	}
	\foreach \p in {0,...,20}
	\foreach \inc in {0,12,18}
	{
		\pgfmathtruncatemacro\l{mod(\p+\inc,21)}
		\draw (point\p) -- (line\l);
	}
	\foreach \p in {0,...,20}
	\foreach \inc in {10,17}
	{
		\pgfmathtruncatemacro\l{mod(\p+\inc,21)}
		\draw[red,dashed] (point\p) -- (line\l);
	}
	\end{tikzpicture}
	\end{center}
	\caption{A layout for $G_4$ based on subplanes of $P_2(\mathbb{F}_4)$.}
	\label{fig:G_4}
\end{figure*}

\subsection{Modified Incidence Graph of Finite Projective Planes}\label{subsec:brown}

In the previous graph $G_q$, each vertex $(0, P)$ can be identified with its pair $(1, P)$, for every $P \in P_2(\mathbb{F}_q)$, giving a graph of diameter 2 very close to the Moore bound. Independently and simultaneously, Brown in \cite{Brown} and Erd\H{o}s \latin{et al.} in \cite{Erdos_hungarian} defined this graph, which is introduced next. Interestingly, Brahme \latin{et al.} have recently unknowingly reinvented these graphs with a different construction and in \cite{Brahme} they already proposed them for HPC clusters. However, in this paper the next definition will be considered as the network topology model.

\begin{definition}\label{def:Brown} Let $q$ be a power of a prime number. Let $\overline{G}_q = (V, E)$ be the graph with vertex set	$$V=P_2(\mathbb{F}_q)$$ and set of adjacencies $$E=\{ \{P,L\} \mid P\perp L,\ P \neq L,\ P,L\in P_2(\mathbb F_q)\}.$$
\end{definition}

Clearly, $\overline{G}_q $ has $q^2+q+1$ vertices. Now, since $P_2(\mathbb{F}_q)$ contains $q+1$ points $X$ such that $X\perp X$, this graph
is a non-regular graph with degrees $q$ and $q+1$. Hence, its number of vertices is $N=q^2+q+1=\Delta^2-\Delta+1$, where $\Delta=q+1$ is the maximum degree. Note that this expression is very close to the Moore bound $M(\Delta,2)=\Delta^2+1$. In the next example it is shown how $\overline{G}_2$ is obtained from $G_2$.

\begin{example}
In Figure \ref{fig:Brown} the graph $\overline{G}_2$ is represented. Note that this is the modified incidence graph obtained from $G_2 ,$ which was considered in Example \ref{ex:Fano}. Therefore, vertex 111, which is obtained identifying point and line 111 in $G_2$, is adjacent to 110, since point and line 110 where adjacent in $G_2$ to 111.

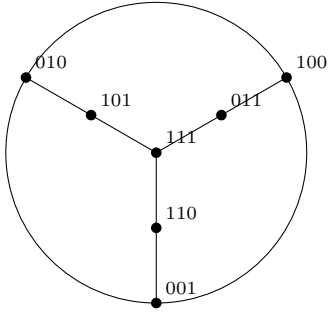
\begin{figure}
	\begin{center}
	\begin{tikzpicture}
	\begin{scope}[
			vertex/.style={anchor=south west,font=\scriptsize},
			]
		\fill (0,0) circle (2pt) node[vertex] {111}
			(30:1)  circle (2pt) node[vertex] {011}
			(30:2)  circle (2pt) node[vertex] {100}
			(150:1) circle (2pt) node[vertex] {101}
			(150:2) circle (2pt) node[vertex] {010}
			(270:1) circle (2pt) node[vertex] {110}
			(270:2) circle (2pt) node[vertex] {001}
			;
		\draw (0,0) circle (2)
			(0,0) -- (30:2)
			(0,0) -- (150:2)
			(0,0) -- (270:2)
			;
	\end{scope}
	\end{tikzpicture}
	\end{center}
	\caption{Modified incidence graph $\overline{G}_2$.}
	\label{fig:Brown}
\end{figure}

\end{example}

\begin{lemma} For each pair of vertices of $\overline{G}_q$ there is a unique minimum path.
\end{lemma}

\begin{proof} Let $P, Q$ be two vertices in $\overline{G}_q$. If $P$ and $Q$ are adjacent, straightforwardly there is a unique edge joining them. On the contrary, if
they are not adjacent, their vector product is adjacent to both, which gives a minimum path between them. If any other minimum path were exist, the two paths will form a square in the graph, which is not possible.

The nonexistence of a square can be proved as follows. Let the points $P$, $Q$ be adjacent to the points $X$ and $Y$. Let $C$ be the cross point of the lines $PQ$ and $XY$. Point $C$ is adjacent to $P$ and $Q$, since it is a linear combination of $X$ and $Y$. In the same way it is adjacent to $X$ and $Y$. Furthermore, $C$ is adjacent to all the points in the lines $PQ$ and $XY$, and hence to all the points in the plane, which contradicts the maximum degree being $q+1$.
\end{proof}

\begin{theorem} The link utilization of $\overline{G}_q$ is $u = \frac{2q^2+q+1}{2q(q+1)}.$
\end{theorem}

\begin{proof} The vector product of a vertex of degree $q$ and a vertex of degree $q+1$ is the vertex of degree $q$. It follows that there is no pair of adjacent vertices of degree $q$, since both should be their vector product. Thus, there are two types of edges: edges with endpoint degrees $q$--$(q+1)$ and edges with endpoint degrees $(q+1)$--$(q+1)$. The remainder of the proof consists on counting the amount of traffic over these links and their number.

First, let us consider edges of type $q$--$(q+1)$. Thus, let us denote $X$ the vertex of degree $q$ and $Y$ the vertex of degree $q+1$. There are $q+1$ vertices of degree $q$ and for each of these vertices there are $q$ edges, all of this type. Therefore, there are $q(q+1)$ vertices of this type.
The traffic traversing the arc from $X$ to $Y$ is composed from the traffic from: 1 path from $X$ to $Y$, $q-1$ paths from neighbours of $X$ to $Y$, and $q$ paths from $X$ to neighbours of $Y$; which gives a total of $2q$ paths.

Next, let us consider edges of type $(q+1)$--$(q+1)$. Let us denote the endpoints $X$ and $Y$.
The total number of edges in $\overline{G}_q$ is $\frac{q(q+1)+(q+1)q^2}{2}=\frac{q(q+1)^2}{2}.$
The number of edges of this type is then \begin{multline*}
	\frac{q(q+1)^2}{2}-q(q+1)=q\frac{(q^2+2q+1)-(2q+2)}{2}\\
	=\frac{q(q^2-1)}{2}.
\end{multline*}
The vertices $X$ and $Y$ have a common neighbour $X\times Y$, whose traffic does not go through this edge. Thus, the traffic from $X$ to $Y$ is due to: 1 path from $X$ to $Y$, $q-1$ paths from neighbors of $X$ to $Y$, and $q-1$ paths from $X$ to neighbours of $Y$; which constitute a total of $2q-1$ paths.

The maximum load is therefore on $q$--$(q+1)$ links. The average use of the links can be calculated as follows: $$\frac{(2q)(q(q+1))+(2q-1)\frac{q(q^2-1)}{2}}{\frac{q(q+1)^2}{2}}=\frac{2q^2+q+1}{q+1}.$$
Finally, the average link utilization at the saturation point is equal to the average use between the maximum use, this is,	$$u=\frac{\frac{2q^2+q+1}{q+1}}{2q}=\frac{2q^2+q+1}{2q(q+1)}.$$

\end{proof}

\begin{notation}
Previous families of graphs constitute the topological models of Projective Networks. We will refer to PN when the graph $G_q$ is considered, and to demi-PN when the graph $\overline{G}_q$ is selected.
\end{notation}

%% file: topologies.tex
\section{Topologies Near the Moore Bound}\label{sec:topologies}

As stated in previous sections, our aim is to find topologies being optimal according to Equations \eqref{eq:cost} and \eqref{eq:final}. That is, for a given $\bar k$ and $R$, the goal is to find well-balanced topologies with maximum number of terminals $T$. Thus, in Subsection~\ref{subsec:peques}, topologies with small average distance are considered, that is, $\overline{k} \leq 2$. The MMS graph has been proposed for interconnection networks with the name of Slim Fly and for this reason it is analyzed in depth in Subsection \ref{subsec:slimfly}. Although the MMS graph is a generalized Moore graph with diameter 2 and $\overline{k} = 2$, its link utilization converges to $8/9$, so it does not reach the bound in Equation~\eqref{eq:final}. In Subsection \ref{subsec:grandes} some other projective constructions of a greater average distance than the ones presented in Section \ref{sec:projective} are summarized. In Subsection \ref{subsec:random} random graphs are considered since they are close to the Moore bound.

\subsection{Topologies with Small Average Distance}\label{subsec:peques}

In this subsection graph constructions approaching the generalized Moore bound and average distance between 1 and 2 are considered. Straightforwardly, the only graphs with $\bar k=1$ are the complete graphs,
which are indeed Moore graphs. As stated in previous section, complete graphs are the optimal topologies as long as routers with enough radix are available. There are many other generalized Moore graphs with $\bar k$ between 1 and 2, for example: the Turán graph, the Paley graph and the Hamming graph of dimension 2, which are described next. Some small examples are shown in Figure~\ref{fig:peques}.

The \firstuse{Turán graph}~\cite{Chartrand} Turán($n$,$r$) is a complete multipartite graph on $n$ vertices.
Let $s_1,\dotsc,s_r$ be $r$ subsets of $\{1,\dotsc,n\}$ with cardinal number $\lfloor n/r\rfloor$ or $\lceil n/r \rceil$. Then, two vertices are connected if and only if they are in different subsets.
Note that the Turán graph contains the complete bipartite graph as a special case:
	$$\text{Turán}(2n,2)\cong K_{n,n}.$$ In the limit the Turán graph has average distance $\lim_{N\rightarrow\infty}\bar k=1+\frac{1}{r}=1.5, 1.\bar 3, 1.25, 1.2, 1.1\bar 6, \dots$

The \firstuse{Paley graph}~\cite{Bollobas} is a graph with $\lim_{N\rightarrow\infty}\bar k=1.5$ very similar to the complete bipartite graph. Let $q$ be a prime power satisfying $q\equiv 1 \pmod 4$. Then, the Paley graph $\text{Paley}(q)$ is the graph whose vertices are the elements of the finite field of $q$ elements $\mathbb F_q$. Two vertices $a,b\in\mathbb F_q$ are connected in $\text{Paley}(q)$ if the difference $a-b$ has its square root in $\mathbb F_q$, \latin{i.e.,} if there is $x\in\mathbb F_q$ such that $a-b=x^2$. A notable property of this graph is that it is \firstuse{self-complementary}: it is isomorphic to the graph that connects vertices if they are not connected in the Paley graph. The Paley graph will appear again later as subgraph of the MMS graph (yet to be introduced).

The \firstuse{Hamming graph}~\cite{Mulder} of side $n$ and dimension 2 is defined as the Cartesian graph product of two complete graphs, $K_n\square K_n$. It is called Hamming graph since two vertices are adjacent if their Hamming distance is 1. In recent networking literature is known as flattened butterfly~\cite{Kim_flat_ISCA}; other names the Hamming graph has received are rook's graph, generalized hypercube~\cite{Bhuyan} and K-cube~\cite{LaForge}.
It has diameter $k=2$, average distance $\bar k=2-\frac{2}{n}-\frac{1}{n^2}$ and size $N=n^2=\Delta^2/4+\Delta+1$, so it is a factor $1/4$ from being asymptotically a Moore graph. Nevertheless, it is a generalized Moore graph, which can result paradoxical; but it can be seen that, although the average distance tends to 2 as a Moore graph would, it is always smaller.

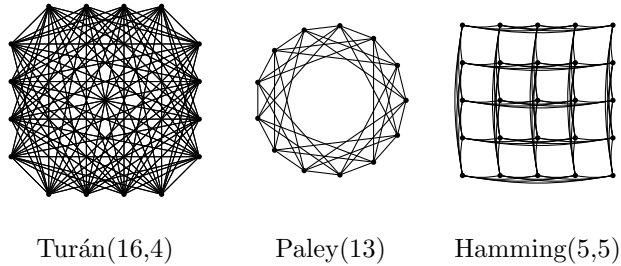
\begin{figure}
    \begin{center}
    \begin{tikzpicture}
    \begin{scope}[x=.5cm,y=.5cm]
        \foreach \a in {0,...,3}
        {
            \begin{scope}[rotate=90*\a]
            \foreach \b in {0,...,3}
            {
                \fill (\b-1.5,2.5) circle (1pt) coordinate (point\a\b);
            }
            \end{scope}
        }
        \foreach \a in {0,...,3}
        \foreach \b in {0,...,3}
        \foreach \c in {1,...,3}
        \foreach \d in {0,...,3}
        {
            \pgfmathtruncatemacro\r{mod(\a+\c,4)}
            \draw (point\a\b) -- (point\r\d);
        }
        \node at (0,-2cm) {Turán(16,4)};
    \end{scope}
    \begin{scope}[xshift=3cm]
        \foreach \a in {0,...,12}
        {
            \fill (\a*360/13:1) circle (1pt) coordinate (point\a);
        }
        \foreach \hop in {1,3,4}
        \foreach \a in {0,...,12}
        {
            \pgfmathtruncatemacro\r{mod(\a+\hop,13)}
            \draw (point\a) -- (point\r);
        }
        \node at (0,-2cm) {Paley(13)};
    \end{scope}
    \begin{scope}[xshift=5.75cm,x=.5cm,y=.5cm]
        \foreach \a in {0,...,4}
        \foreach \b in {0,...,4}
        {
            \fill (\a-2,\b-2) circle (1pt) coordinate (point\a\b);
        }
        \foreach \a in {0,...,4}
        \foreach \b in {0,...,4}
        \foreach \hop in {1,...,4}
        {
            \pgfmathtruncatemacro\r{mod(\a+\hop,5)}
            \pgfmathtruncatemacro\cond{\a<\r}
            \ifthenelse{\cond=1}
            {
                \draw (point\a\b) edge[out=-10,in=-170] (point\r\b);
            }
        }
        \foreach \a in {0,...,4}
        \foreach \b in {0,...,4}
        \foreach \hop in {1,...,4}
        {
            \pgfmathtruncatemacro\r{mod(\b+\hop,5)}
            \pgfmathtruncatemacro\cond{\b<\r}
            \ifthenelse{\cond=1}
            {
                \draw (point\a\b) edge[out=100,in=-100] (point\a\r);
            }
        }
        \node at (0,-2cm) {Hamming(5,5)};
    \end{scope}
    \end{tikzpicture}
    \end{center}
    \caption{The Turán graph, the Paley graph and the Hamming graph}
    \label{fig:peques}
\end{figure}

\subsection{Slim Fly}\label{subsec:slimfly}

Slim Fly is the name given by Besta and Hoefler~\cite{Besta} to network topologies based on the McKay--Miller--\v{S}ir\'{a}\v{n} (MMS) graphs~\cite{MMS}. The MMS is a family of graphs of diameter 2 reaching asymptotically $\frac{8}{9}$ of the vertices given by the Moore bound. When degree $\Delta=7$ is considered, the MMS graph coincides with the Hoffman--Singleton graph~\cite{Hoffman}, which is a Moore graph. Thus, for small number of vertices it is a very good option although it gets slightly worse for larger ones. Figure~\ref{fig:MMS_convergence_N}  shows how the number of vertices of the MMS graph converges to $\frac{8}{9}$ the cardinal given by the Moore bound for $k=2$. Note that the graph attaining value 1 in the ordinates is the Hoffman--Singleton graph, which is a Moore graph.

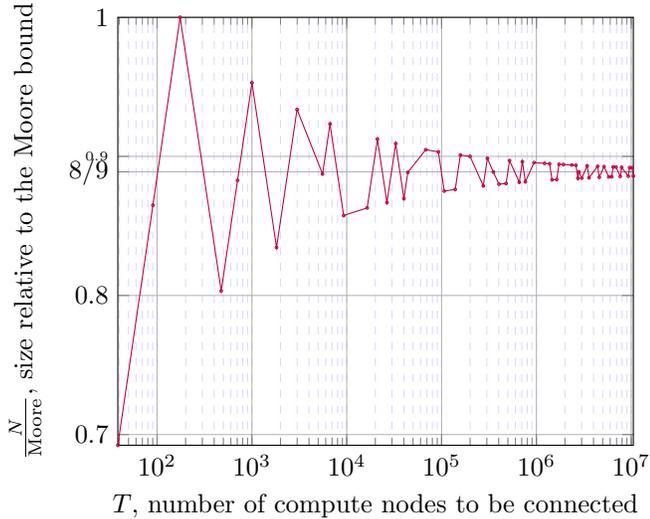
\begin{figure}
	\begin{center}
	\begin{tikzpicture}
	\begin{semilogxaxis}[
		domain=39:1e7,
		enlargelimits=false,
		xmajorgrids=true,
		ymajorgrids=true,
		xminorgrids=true,
		yminorgrids=true,
		minor y tick num=1,
		ytick={0.7,0.8,0.888888888888888,0.9,1.0},
		yticklabels={0.7,0.8,8/9,{\tiny 0.9},1},
		minor grid style={dashed,very thin, color=blue!15},
		major grid style={very thin, color=black!30},
		xlabel={$T$, number of compute nodes to be connected},
		ylabel={$\frac{N}{\text{Moore}}$, size relative to the Moore bound},
		legend style={at={(0.00,1.01)},anchor=south west,font=\scriptsize},legend columns=5,legend cell align=left,
		legend image post style={every path={nomorepostactions}},
		]
	\addplot[purple,mark=o,mark size=.5pt] coordinates { (38.57142857, .6923076923) (90.40000000, .8648648649) (175., 1.) (474.5263158, .8032786885) (705.4545454, .8827586207) (1001.160000, .9529411765) (1818.903226, .8344827586) (2991.756757, .9337016575) (5556.869565, .8873483536) (6658.795918, .9233226837) (9280.981818, .8574821853) (16422.68657, .8629690049) (21070.20548, .9124087591) (26520.83544, .8668252081) (32838.57647, .9091891892) (40087.42857, .8696832579) (44081.02128, .8885032538) (68060.65138, .9048248513) (92538.35537, .9032778076) (1.067178740*10^5, .8750591576) (1.392782446*10^5, .8762395875) (1.577870966*10^5, .9009380863) (1.995821338*10^5, .9000320410) (2.751780229*10^5, .8788184802) (3.040735414*10^5, .8985752234) (3.511026526*10^5, .8887924487) (4.027467638*10^5, .8800235248) (4.791578010*10^5, .8805240175) (5.207439862*10^5, .8969870392) (6.597932085*10^5, .8813726875) (7.111198382*10^5, .8961890452) (7.650415790*10^5, .8817355689) (9.430174679*10^5, .8955342001) (1.220532875*10^6, .8949871588) (1.377677246*10^6, .8947460749) (1.461070098*10^6, .8831266128) (1.637817135*10^6, .8833423347) (1.731299320*10^6, .8943168988) (1.928801024*10^6, .8941250613) (2.367745764*10^6, .8937793786) (2.610212802*10^6, .8936231055) (2.737418987*10^6, .8842168741) (2.802543246*10^6, .8888647769) (3.004096690*10^6, .8843597011) (3.435736580*10^6, .8932089659) (3.588305431*10^6, .8846206676) (4.419196357*10^6, .8928614518) (4.599431876*10^6, .8849602174) (5.169405098*10^6, .8926592547) (5.784622320*10^6, .8852497251) (6.220756024*10^6, .8853369734) (6.446811543*10^6, .8923918138) (6.915219913*10^6, .8923109031) (7.659535802*10^6, .8855753020) (7.919017987*10^6, .8921597996) (9.304635580*10^6, .8857836591) (9.599848432*10^6, .8919566102) (1.020887547*10^7, .8918943764) (1.052281765*10^7, .8859085926) };
	\end{semilogxaxis}
	\end{tikzpicture}
	\end{center}
	\caption{Convergence on the number of vertices in the MMS graph to $\frac{8}{9}$ of Moore bound for diameter 2.}
	\label{fig:MMS_convergence_N}
\end{figure}

Let us now give a schematic definition of this graph based on the ideas in \cite{Hafner}. Let $q$ be a prime power other than 2.
Then, for some $\varepsilon\in\{-1,0,1\}$, $q\equiv \varepsilon\pmod 4$.
As $q$ is a prime power there is a (unique) finite field of $q$ elements, which is denoted by $\mathbb F_q$.
The set of vertices is defined as
	$$V(\mathrm{MMS}(q))=\{(s,x,y) \mid s\in\{0,1\},\ x,y\in\mathbb F_q\}.$$
Thus, MMS($q$) is a graph with $2q^2$ vertices.
In order to define the set of adjacencies a \textsl{primitive element} $\xi\in\mathbb F_q$ has to be found, that is, an element $\xi$ satisfying $\{\xi^i\mid i\in\mathbb Z\}=\mathbb F_q\setminus \{0\}$. This implies that $\xi^{q-1}=1$. Now, let us first define the sets

	$$X_0=\begin{cases}
	\{1,\xi^2,\dotsc,\xi^{q-3}\}	&\text{ if $\varepsilon=1$,}\\
	\{1,\xi^2,\dotsc,\xi^{\frac{q-1}{2}},\xi^{\frac{q+1}{2}},\dotsc,\xi^{q-2}\}	&\text{ if $\varepsilon=-1$,}\\
	\{1,\xi^2,\dotsc,\xi^{q-2}\}	&\text{ if $\varepsilon=0$,}\\
	\end{cases}$$
and $X_1=\xi X_0$.
Later it will be used that $|X_0|=\frac{q-\varepsilon}{2}$, $X_0\cup X_1=\mathbb F_q\setminus \{0\}$ and
	$$X_0\cap X_1=\begin{cases}
	\emptyset	&\text{ if $\varepsilon=1$,}\\
	\{1,-1\}	&\text{ if $\varepsilon=-1$,}\\
	\{1\}		&\text{ if $\varepsilon=0$.}\\
	\end{cases}$$
The adjacencies are defined as follows:
\begin{enumerate}
\item $(s,x,y_1)$ is adjacent to $(s,x,y_2)$ for all $s\in\{0,1\}$, $x,y_1,y_2\in\mathbb F_q$ such that $y_1-y_2\in X_s$.
\item $(0,x_1,y_1)$ is adjacent to $(1,x_2,y_2)$ for all $x_1,x_2,y_1,y_2\in\mathbb F_q$ such that $y_1-y_2=x_2x_1$.
\end{enumerate}
Thus, each vertex has $|X_0|$ incident edges by the first item and $q$ incident edges by the second item. Therefore, the degree of MMS($q$) is $\Delta=\frac{3q-\varepsilon}{2}$.
For convenience, let us call the edges by item 1), \textsl{local edges} and the edges by item 2), \textsl{global edges}.

The MMS has diameter 2. Let us study the minimum paths to prove this, and further, to count the use of local and global edges. The possible routes between two vertices could be \emph{ll}, \emph{lg}, \emph{gl} or \emph{gg}; where \emph{l} means a local edge and \emph{g} a global edge.
Let $(s_1,x_1,y_1)$ be the origin vertex and $(s_2,x_2,y_2)$ the destination.
If $s_1=s_2$ and $x_1=x_2$ then the minimum routes are \emph{ll}; this is the same that in Paley graphs. Half of the vertices $(s_1,x_1,y_m)$ can be used as the middle vertex.
If $s_1=s_2$ but $x_1\neq x_2$ then the minimum route is \emph{gg} with some middle vertex $(1-s_1,x_m,y_m)$.
The adjacency exists if $y_1-y_m=(1-2s_1)x_mx_1$ and $y_2-y_m=(1-2s_1)x_2x_m$.
Hence, the vertex in the middle is unique and can be calculated by $x_m=(1-2s_1)(y_1-y_2)/(x_1-x_2)$ and $y_m=y_1-(1-2s_1)x_mx_1$.
If $s_1=1-s_2=s$ then the minimum routes will be half of the time \emph{lg} and the other half \emph{gl}.
The equations for a middle vertex $(s,x_1,y_m)$ are $y_m=y_2+(1-2s)x_1x_2$ and $z=y_1-y_2-(1-2s)x_1x_2\in X_s$, while that for a middle vertex $(1-s,x_2,y_m)$ they are $y_m=y_1-(1-2s)x_1x_2$ and $z=y_1-y_2-(1-2s)x_1x_2\in X_{1-s}$. Thus, routing is performed by computing $z=y_1-y_2-(1-2s)x_1x_2$. If $z=0$ there is a global edge from the origin to the destination, otherwise, as $X_s\cup X_{1-s}=\mathbb F_q\setminus \{0\}$, either $z\in X_s$ or $z\in X_{1-s}$. If $z\in X_s$ use the middle vertex $(s,x_1,y_m)$ and if $z\in X_{1-s}$ use the middle vertex $(1-s,x_2,y_m)$. The uniqueness depends, therefore, in $X_s\cap X_{1-s}$; if $\varepsilon=1$ then it is always the empty set and the route is unique, otherwise there are some pairs for which there are two minimal paths.
As summary, the number of routes \emph{gg} is asymptotically the sum of the number of routes \emph{lg} plus routes \emph{gl}. Thus, 3 global links are used per each local link used.

The analysis in \cite{Besta} does not consider the link utilization and concludes that $\Delta_0=\frac{\Delta}{2}$ terminals per router are required for a full use of the network. As studied in Section~\ref{sec:model}, this would be true if all links would accept the same load. However, this is not the case in the MMS as shown next. As proved above, the number of global links is about 2 times the number of local links, but the load over the total of global links is about 3 times the load of the local links. Thus, each global link receives about $3/2$ of the load received by a local link. Hence, saturation is reached when global links receive load 1 and local links receive $2/3$.
Then, the link utilization is $u=\frac{2}{3}\cdot 1+\frac{1}{3}\cdot \frac{2}{3}=\frac{8}{9}$.\footnote{The value $8/9$ is the same that the quotient of its number of vertices to the Moore bound. This is a coincidence, it does not hold in the great majority of graphs.}

Figure~\ref{fig:MMS_convergence_u} shows this convergence of the link utilization to $\frac{8}{9}$. Again, note that this is an asymptotic behaviour; for the case $q=5$---the Hoffman--Singleton graph---all links receive the same load and the utilization is $u=1$ since it is a symmetric graph. The situation is a little worse if $\varepsilon \neq 1$, where there are non-unique minimal paths and, if the routing is deterministic, there are a few links that are used exclusively for messages between their endpoints.

\begin{figure}
	\begin{center}
	\begin{tikzpicture}
	\begin{semilogxaxis}[
		domain=39:1e7,
		enlargelimits=false,
		xmajorgrids=true,
		ymajorgrids=true,
		xminorgrids=true,
		yminorgrids=true,
		minor y tick num=1,
		ytick={0.888888888888888,0.9,0.95,1.0},
		yticklabels={8/9,0.9,0.95,1},
		minor grid style={dashed,very thin, color=blue!15},
		major grid style={very thin, color=black!30},
		xlabel={$T$, number of compute nodes to be connected},
		ylabel={$u$, average link utilization},
		legend style={at={(0.00,1.01)},anchor=south west,font=\scriptsize},legend columns=5,legend cell align=left,
		legend image post style={every path={nomorepostactions}},
		]
	\addplot[purple,mark=o,mark size=.5pt] coordinates { (38.57142857, .8571428571) (90.40000000, .9416666667) (175., 1.) (474.5263158, .8803827751) (705.4545454, .9185606061) (1001.160000, .9507692308) (1818.903226, .8842504744) (2991.756757, .9317211949) (5556.869565, .9044384058) (6658.795918, .9216326531) (9280.981818, .8865203762) (16422.68657, .8869936034) (21070.20548, .9111440207) (26520.83544, .8873108984) (32838.57647, .9080711354) (40087.42857, .8875379939) (44081.02128, .8968306738) (68060.65138, .9039199333) (92538.35537, .9024522422) (1.067178740*10^5, .8879466990) (1.392782446*10^5, .8880332354) (1.577870966*10^5, .9002361833) (1.995821338*10^5, .8993791825) (2.751780229*10^5, .8882182986) (3.040735414*10^5, .8980025499) (3.511026526*10^5, .8929002193) (4.027467638*10^5, .8883029006) (4.791578010*10^5, .8883376888) (5.207439862*10^5, .8965036148) (6.597932085*10^5, .8883962096) (7.111198382*10^5, .8957511745) (7.650415790*10^5, .8884210526) (9.430174679*10^5, .8951340616) (1.220532875*10^6, .8946187806) (1.377677246*10^6, .8943917626) (1.461070098*10^6, .8885152884) (1.637817135*10^6, .8885297611) (1.731299320*10^6, .8939877301) (1.928801024*10^6, .8938071743) (2.367745764*10^6, .8934818873) (2.610212802*10^6, .8933348626) (2.737418987*10^6, .8885880451) (2.802543246*10^6, .8909045048) (3.004096690*10^6, .8885975048) (3.435736580*10^6, .8929453158) (3.588305431*10^6, .8886147461) (4.419196357*10^6, .8926185318) (4.599431876*10^6, .8886370962) (5.169405098*10^6, .8924284353) (5.784622320*10^6, .8886560784) (6.220756024*10^6, .8886617857) (6.446811543*10^6, .8921770438) (6.915219913*10^6, .8921009985) (7.659535802*10^6, .8886773443) (7.919017987*10^6, .8919589935) (9.304635580*10^6, .8886909084) (9.599848432*10^6, .8917680641) (1.020887547*10^7, .8917095911) (1.052281765*10^7, .8886990248)};
	\end{semilogxaxis}
	\end{tikzpicture}
	\end{center}
	\caption{Convergence on the link utilization in the MMS graph to $\frac{8}{9}$.}
	\label{fig:MMS_convergence_u}
\end{figure}
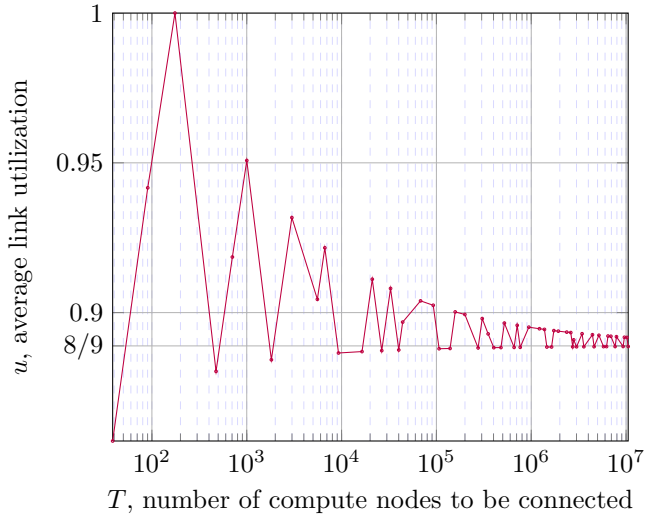

\subsection{Projective Networks of Higher Average Distance} \label{subsec:grandes}

In Section~\ref{sec:projective} two projective networks of average distances 2 and 2.5 were presented. There are also graphs based on projective spaces which attain the bounds for greater average distances. In this subsection they are enumerated. They are not described in a great detail since such an amount of terminal nodes is beyond the horizon of current network topologies.

The incidence graph over a generalized quadrangle or hexagon, instead of the projective plane, results in a generalized Moore graph with average distance tending to 3.5 and 5.5 respectively~\cite{Exoo}. Alike happens to $G_q$, generalized quadrangles and hexagons exist whenever $q$ is a prime power. Their number of vertices is the double of the number of points in their spaces, respectively $P_3(\mathbb F_q)$ and $P_5(\mathbb F_q)$.

Furthermore, these graphs allow for a modification similar to $\overline{G}_q$, as it was proved by Delorme~\cite{Delorme_diametro3}. In the case of quadrangles the resulting average distance tends to 3 and for hexagons it tends to 5. In both cases, the number of vertices is asymptotically close to the Moore bound. However $q$ must be an odd power of 2. Hence, they exist only for a very reduced amount of sizes. Otherwise, Delorme's graph on quadrangles, that is the modified incidence graph on the quadrangles over $P_3(\mathbb F_q)$, would have been a very good alternative to current dragonfly topology.
These graphs are denoted as \firstuse{Delorme's graph} in the remainder of the paper. By default this notation will refer to the construction using generalized quadrangles, unless specified otherwise.

\subsection{Random Graphs}\label{subsec:random}

\firstuse{Random graphs}~\cite{Bollobas} have been proposed for interconnection networks of datacenters~\cite{Jellyfish} and HPC~\cite{Koibuchi_random}. Since not many generalized Moore graphs are known, random graphs might constitute an alternative when specific constructions are not known. There are three major different models to define a random graph with $N$ vertices. Each one of these models requires a different additional parameter: a probability $p$ of each edge (the \textsl{binomial model}), a total number of edges $M$ (the \textsl{uniform model}), or a constant degree $\Delta$. Although they are very similar when $\Delta=p(N-1)=2M/N$, the three models are pairwise different. Nevertheless, for all our purposes the approximations work equally fine indistinctly of which model is chosen. The average distance is approximately $\bar k\approx \frac{\log T}{\log R}-1$ which is close to the Moore bound for all $\bar k$, although worse than the values for specific known constructions. Thus, random graphs could be used if there is no appropriate construction for the desired dimension. Furthermore, the link utilization in random graphs is a delicate aspect. If all terminals generate the same amount of traffic, then experimentally we have obtained an utilization of $u\approx 0.8$ (depending on the model), lower than all the topologies considered in this paper.

%% file: comparative.tex
\section{Comparison of the Topologies}\label{sec:comparative}


In this section a comparison of the topologies presented in previous Sections \ref{sec:projective} and \ref{sec:topologies} is done in terms of the cost model presented in Section \ref{sec:model}. The section is divided into three subsections. The first one considers the complete picture of all the networks with diameters from 1 to 6. In such subsection also other topologies such as the dragonfly~\cite{Kim_dgfly_ISCA}, 3D Hamming graph and Hypercube are also considered as useful references. The second subsection is focused on a detailed comparison among projective networks and Slim Fly. Finally, the third subsection considers different implementations for two specific numbers of compute nodes, which are 10,000 and 25,000.

\subsection{General Comparison}\label{subsec:comparison}

\begin{table}
	\begin{center}\scriptsize
	\begin{tabular}{|lccc|}
	\hline
	\rule{0pt}{2.5ex}Graph									&	$k$		&	$\lim_{N\rightarrow \infty}\bar k$	&	$\lim u$	\\
	\hline
	Complete graph $K_N$									&	1		&	1									&	1			\\
\hline
	Turán($N$,$r$)											&	2		&	$1+\frac{1}{r}$						&	1			\\
	Complete bipartite graph $K_{n,n}$						&	2		&	1.5									&	1			\\
	Hamming graph 2D										&	2		&	2									&	1			\\
	Demi-projective network $\overline{G}_q$				&	2		&	2									&	1			\\
	Slim Fly MMS for $q=4w+\varepsilon$						&	2		&	2									&	$8/9$		\\
\hline
	Projective network $G_q$								&	3		&	2.5									&	1			\\
	Dragonfly												&	3		&	3									&	1			\\
	Delorme's graph on quadrangles							&	3		&	3									&	1			\\
	Hamming graph 3D										&	3		&	3									&	1			\\
\hline
	Incidence graph of generalized quadrangles				&	4		&	3.5									&	1			\\
	Delorme's graph on hexagons								&	5		&	5									&	1			\\
	Incidence graph of generalized hexagons					&	6		&	5.5									&	1			\\
	Random graph with $N$ vertices							&	\hbox to 0pt{\hss$\sim\frac{\log(N)}{\log{\Delta}}$\ \hss}	&	$\sim\frac{\log(N)}{\log{\Delta}}$	&	$\approx$0.8\\
\hline
	Hypercube $C_2^n$										&	$n$		&	$n/2$								&	1			\\
	\hline
	\end{tabular}
	
	\end{center}
	\caption{Topological parameters of optimal topologies and some references.}
	\label{tbl:topologies}
\end{table}

\begin{table*}
	\begin{center}\scriptsize
	\renewcommand\arraystretch{1.1}
	\begin{tabular}{|lccccc|}
	\hline
	Graph																&	$T$							&	$R$						&	$N$				&	$\Delta$				&	$\Delta_0$				\\
	\hline
	Complete graph $K_N$												&	$N^2$						&	$2N-1$					&	$N$				&	$N-1$					&	$N$						\\
	\hline
	Turán($N$,$r$)														&	$N^2\frac{r-1}{r+1}$		&	$N\frac{(r-1)(2r+1)}{r(r+1)}$&	$N$			&	$N\frac{r-1}{r}$			&	$N\frac{r-1}{r+1}$		\\
	Complete bipartite graph $K_{n,n}$									&	$4n^2/3$					&	$5n/3$					&	$2n$			&	$n$						&	$2n/3$					\\
	Hamming graph 2D of side $n$									&	$n^3$						&	$3n-2$					&	$n^2$			&	$2(n-1)$				&	$n$						\\
	Demi-projective network $\overline{G}_q$								&	$q^3/2+q^2+q+1/2$			&	$3(q+1)/2$				&	$q^2+q+1$		&	$q+1$					&	$(q+1)/2$				\\
	Slim Fly MMS for $q=4w+\varepsilon$											&	$4/9q^2(3q-\varepsilon)$	&	$13/18(3q-\varepsilon)$	&	$2q^2$			&	$(3q-\varepsilon)/2$	&	$2/9(3q-\varepsilon)$		\\
	\hline
	Projective network $G_q$										&	$4/5(q^3+2q^2+2q+1)$		&	$7(q+1)/5$				&	$2(q^2+q+1)$	&	$q+1$					&	$2(q+1)/5$				\\
	Dragonfly with $h$ global links per router							&	$4h^4+2h^2$					&	$4h-1$					&	$4h^3+2h$		&	$3h-1$					&	$h$						\\
	Delorme's graph on generalized quadrangles			    &	$(q+1)^2(q^2+1)/3$			&	$4/3(q+1)$				&	$q^3+q^2+q+1$	&	$q+1$					&	$(q+1)/3$				\\
	Hamming graph 3D of side $n$					&	$n^4$						&	$4n-3$					&	$n^3$			&	$3(n-1)$				&	$n$						\\
	\hline
	Incidence graph of generalized quadrangles							&	$4/7(q+1)^2(q^2+1)$			&	$9/7(q+1)$				&	$2(q^3+q^2+q+1)$&	$q+1$					&	$2(q+1)/7$				\\
	Delorme's graph on generalized hexagons				&	$1/5(q^4+q^2+1)(q+1)^2$		&	$6/5(q+1)$				&	$q^5+\dotsb+q+1$&	$q+1$					&	$(q+1)/5$				\\
	Incidence graph of generalized hexagons								&	$4/11(q^4+q^2+1)(q+1)^2$	&	$13/11(q+1)$			&	$2(q^5+\dotsb+q+1)$&	$q+1$				&	$2(q+1)/11$				\\
	Random graph with $N$ vertices and degree $\Delta$					&	$\Delta\log(\Delta)N/\log(N)$	&	$\Delta(1+\frac{\log \Delta}{\log N})$	&	$N$		&	$\Delta$				&	$\sim \frac{\Delta\log{\Delta}}{\log{N}}$	\\
	\hline
	Hypercube $C_2^n$													&	$2^{n+1}$					&	$n+2$					&	$2^n$			&	$n$						&	2						\\
	\hline
	\end{tabular}
	\end{center}
	\caption{Structural parameter of optimal known topologies and some references.}
	\label{tbl:structural}
\end{table*}

\foreach \R in {64}
{
\begin{figure*}[ht]
	\begin{center}
	\newdimen\crossx
	\includegraphics{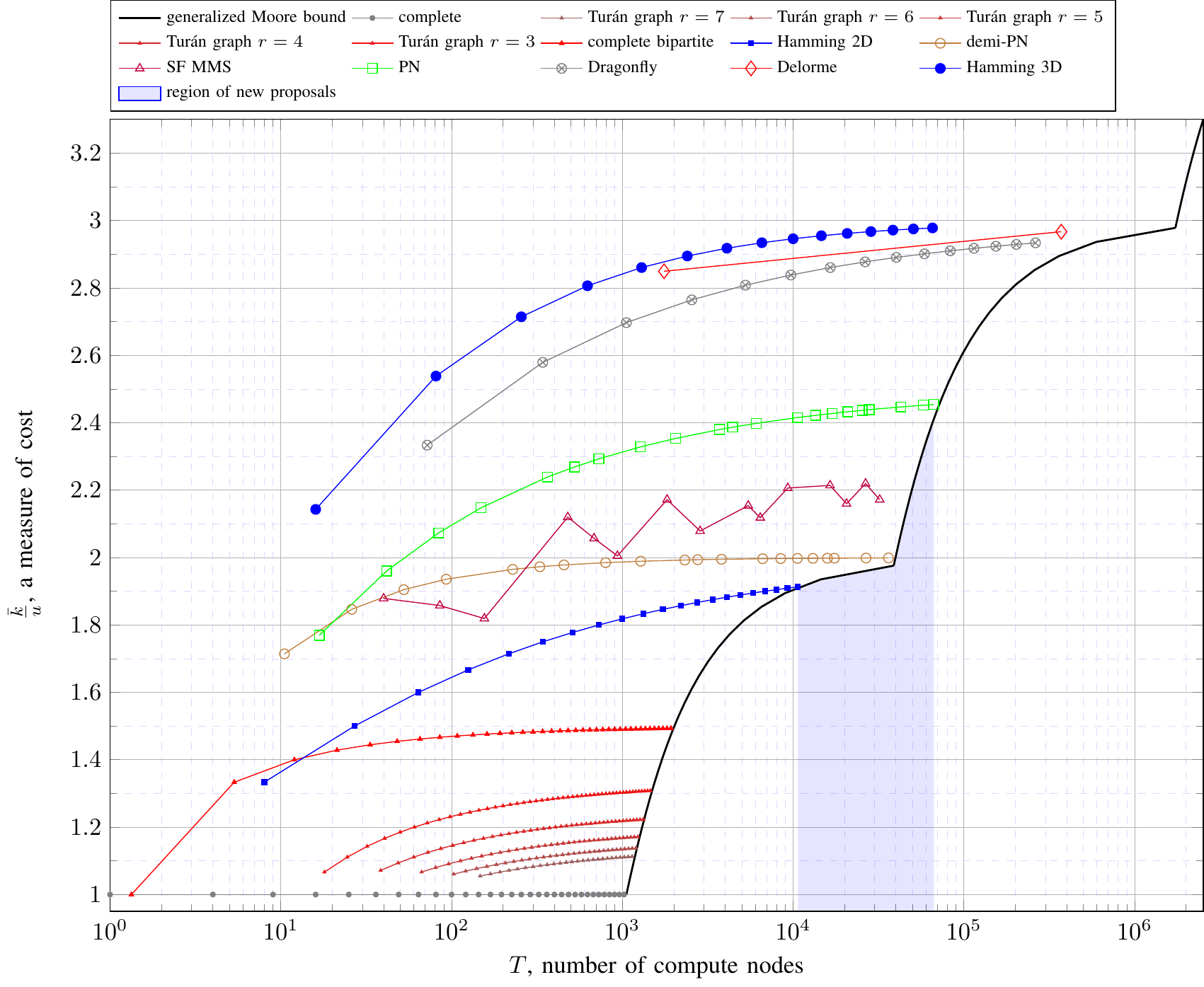}
	\end{center}
	\caption{The measure of cost $\bar k/u$ in realizations of topologies with a given number of compute nodes using routers with maximum radix \R.}
	\label{fig:cost}
\end{figure*}
}

\begin{figure*}[ht]
	\begin{center}
	\newdimen\crossx
	\includegraphics{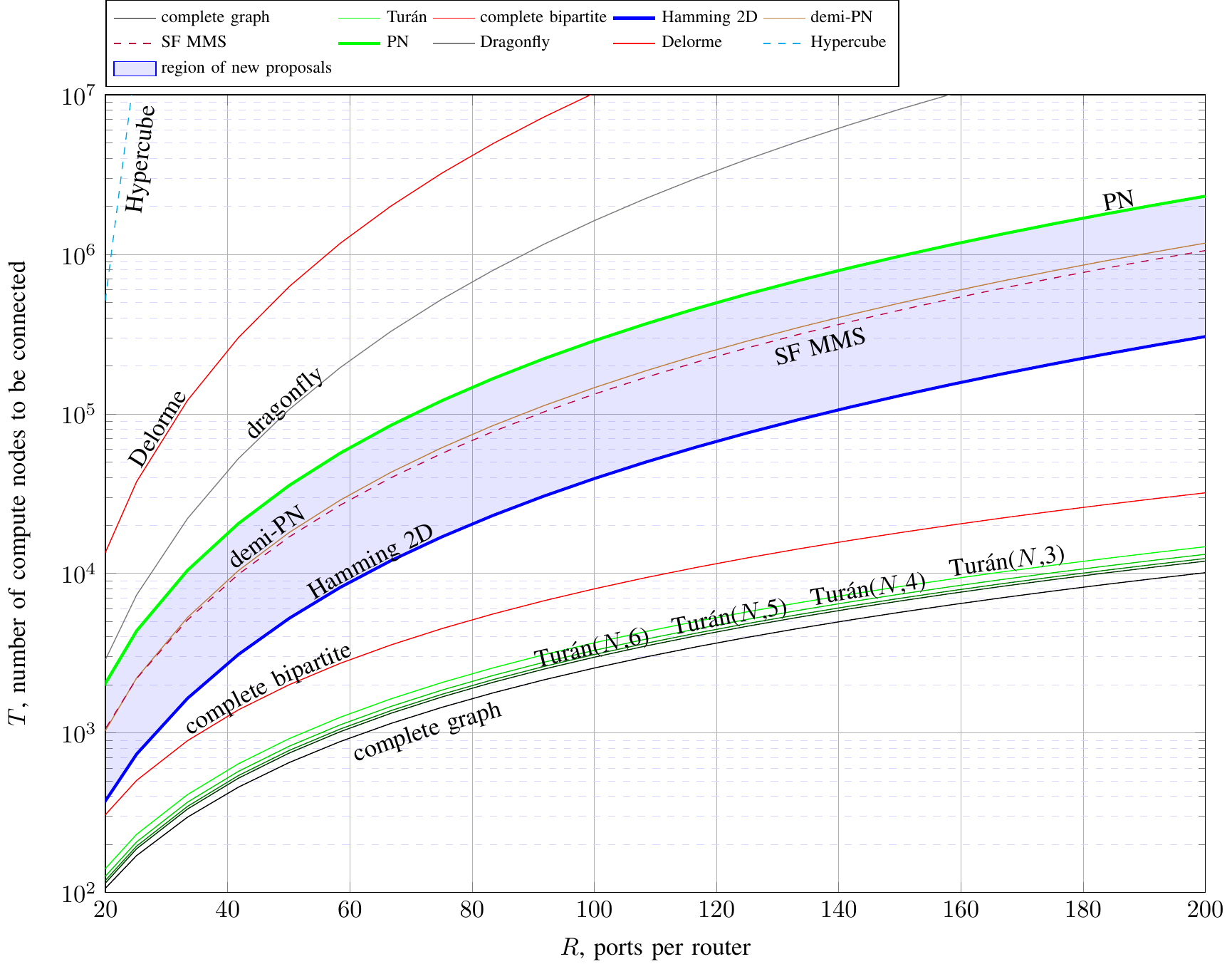}
	\end{center}
    \caption{Scalability of the different topologies.}
	\label{fig:map}
\end{figure*}

Table~\ref{tbl:topologies} summarizes the fundamental parameters of the graphs presented in Section \ref{sec:topologies}: the diameter and the limit values of average distance and utilization.
Table~\ref{tbl:structural} contains the parameters relevant to a network implementing the topology.
Both tables present these values for the optimal graphs, other graphs which are close to be optimal and other graphs, such as the hypercube, to take as a reference.

Figure~\ref{fig:cost} illustrates the cost of networks implementing different topologies using routers with at most 64 ports.
Other values of $R$ give similar plots.
The thick black curve is the average distance corresponding to an ideal generalized Moore graph with $u=1$ (like Equation~\eqref{eq:final}), which is a lower bound for the values of the other curves.
Each other curve corresponds to a topology, which is build for all possible radix up to 64.
The value of $\Delta_0$ has been tried to be a natural number, but sometimes this condition has been relaxed to avoid under/over-subscription, which would distort the figure.
The ordinates axis shows the value $\bar k/u$ which, according to Equation~\eqref{eq:cost}, is a measure of cost associated to the topology. Thus, curves that attain the bound are the optimal topologies, which are: the complete graph, the Turán graphs, the 2D Hamming graph, demi-PN, PN and Delorme's graph $P_3(\mathbb F_q)$. Note that $P_3(\mathbb F_q)$
intersects the curve in the limit. However, it only exists when $\Delta-1$ is a odd power of 2 which means that there are only two points in the range $R\leq 64$.
The MMS graph does not attain the bound because of its asymmetry; as we have seen in previous sections, the MMS has $u=8/9$ in the limit. Hence, the curve is about $9/8$ the one of demi-PN.
For greater average distances the dragonflies do scale very well, although not attaining the bound. As it can be observed the 3D Hamming graph is completely superseded by the dragonfly.

Figure~\ref{fig:map} indicates which topologies are realizable for a given number of terminals $T$ and available router radix $R$. It holds that solid lines are sorted by average distance. Hence, the optimal topology is the solid line immediately above the desired $(R,T)$ point.

\subsection{Projective Networks vs Slim Fly}\label{subsec:discu}

This subsection explains in more detail the advantages of PN and demi-PN with respect to the SF MMS in the design of new high scale interconnection networks. It will be shown that link utilization is an important parameter in the network cost model. For this explanation, Figure~\ref{fig:cost_Brown} will be used. In this figure both curves $\bar{k}$ and $\frac{\bar{k}}{u}$ for the three topologies PN, demi-PN and SF MMS are shown. Note that for PN both curves coincide since the graphs $G_q$ are symmetric, as it has been proved in Theorem~\ref{thm:symmetry}.

Clearly, if only average distance is considered, the smaller cost is given by SF MMS. However, its maximum size is $\frac{8}{9}$ smaller than the possible one, which is attained by the demi-PN construction. The reader should notice that the abscises axis is logarithmic, therefore this difference seems smaller in the figure. However, if the link utilization is considered in the network cost model, for more than $1000$ compute nodes demi-PN exhibits as the best alternative both in scalability and cost.

Finally, PN is an alternative to scale to a larger amount of compute nodes reaching almost $10^5$ compute nodes with the minimum cost.

\begin{figure}[]
	\begin{center}
	\includegraphics[]{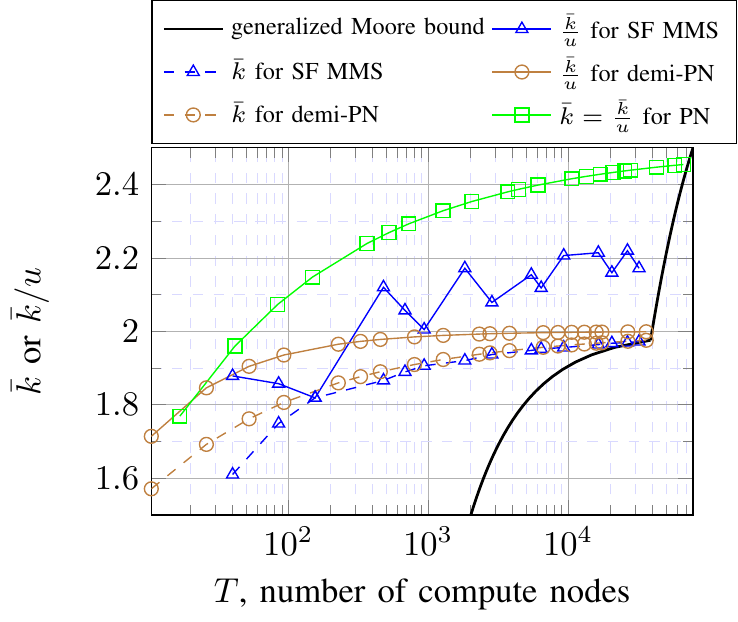}
	\end{center}
	\caption{The measure of cost $\bar k/u$ and $\bar k$ given a number of terminals for SF~MMS, PN and demi-PN. Using routers with maximum radix 64.}
	\label{fig:cost_Brown}
\end{figure}

\subsection{Cases of Use}\label{subsec:cases_of_use}

To exemplify the use of the topologies, in this subsection different specific networks that connect a given amount of compute nodes are shown.
Two approximate network sizes have been selected: 10,000 compute nodes and 25,000 compute nodes.
Even for the small case of $T\approx 10,000$, the complete graph would require a router radix of about $R\approx 200$, which is currently unrealistic.
Hence, the topologies to be considered will be the Hamming graph, the demi-PN, the SF MMS, the PN and the dragonfly.
Tables~\ref{tbl:10000} and \ref{tbl:25000} show the network parameters for each of the selected topologies in the small and large cases, respectively.

The calculations assume that nodes are arranged into fully electrical groups and cables outside them are optical.
These groups are the closest possible to 500 compute nodes, while trying to maximize the connections inside a group.
An electrical group size marked with asterisk in the tables indicates the size for most electrical groups, with a few smaller groups.

For a fair comparison, we have employed the cost models from~\cite{Besta} using speeds of 40~Gbps, avoiding the extra costs of 100G routers and cables which are still in their market introduction stage.
An average intra-rack distance of 1m is assumed, from which it is obtained a price of 0.985\$/Gbps for the average electrical cable.
The average length of the optical inter-rack cables is approximately the average distance of a mesh of same dimensions plus 2m of overhead.
In the 10,000 nodes case, an average cost per optical cable of 7.7432\$/Gbps is computed, and in the 25,000 case of 7.9178\$/Gbps.
The cost per router is modelled as $350.4 R-892.3$ \$/router.
The only power considered is the consumed by the SerDes, which is approximated to 2.8 watts per port.

Tables~\ref{tbl:10000} and ~\ref{tbl:25000} show cost and power per node for the topologies studied. The lowest cost and power are obtained in both cases with a 2D Hamming graph. However, its required switches exceed the current limit of 48 available ports, so it could only be built with either slower links or using multi-chip switches with higher latency, as discussed in Section~\ref{sec:model}. Next, we consider designs realizable with full speed and a single switch chip per router. 
With $T\approx 10,000$ nodes, the demi-PN provides the lowest cost and power, 1\% and 7\% respectively lower than SF MMS. For $T\approx 25,000$, a diameter 3 network is required using switches up to 48 ports. In this case, the PN provides the lowest power, 10.9\% less than the dragonfly. A layout of a projective network requires more optical cables when compared with SF MMS or dragonfly, so in this case the cost of the dragonfly is 2.6\% lower because of its reduced number of optical cables.
Note that, for an all-optical system such as PERCS \cite{Arimilli}, projective networks provide significantly better power and cost per node than the alternatives in the tables.





\begin{table*}
	\begin{center}
	\begin{tabular}{|c|ccccc|}
	\hline
	Topology					&	Hamming $K_{22}^2$	&	demi-PN(27)	&	SF MMS(19)	&	PN(23)	&	dragonfly(7)\\
	\hline
	T							&	10648				&	10598		&	9386		&	9954		&	9702\\
	\textbf{R}							&	\textbf{64}					&	\textbf{42}			&	\textbf{42}			&	\textbf{33}			&	\textbf{27}\\
	N							&	484					&	757			&	722			&	1106		&	1386\\
	$\Delta_0$					&	22					&	14			&	13			&	9			&	7\\
	subscription				&	1.002				&	0.999		&	0.991		&	0.921		&	0.994\\
	Size of electrical group	&	484					&	504*		&	494			&	396*		&	490*\\
	Number of groups			&	22					&	22			&	19			&	26			&	20\\
	Electrical cables			&	5082				&	556			&	3971		&	1907		&	8926\\
	Optical cables				&	5082				&	10028		&	6498		&	11365		&	4514\\
	\hline
	\textbf{Cost per node (\$)}	&	\textbf{1145.41}	&\textbf{1282.59}&\textbf{1294.51}	&	\textbf{1546.83}		&	\textbf{1404.42}\\
	\textbf{Power per node (W)}				&	\textbf{8.15}				&	\textbf{8.40}		&	\textbf{9.05}		&	\textbf{10.27}		&	\textbf{10.80}\\
	\hline
	\end{tabular}
	\end{center}
	\caption{Example networks with about 10,000 compute nodes and electrical groups of about 500 nodes.}
	\label{tbl:10000}
\end{table*}

\begin{table*}
	\begin{center}
	\begin{tabular}{|c|ccccc|}
	\hline
	Topology					&	Hamming $K_{29}^2$	&	demi-PN(37)	&	SF MMS(27)	&	PN(31)		&	dragonfly(9)\\
	\hline
	T							&	24389				&	26733		&	26244		&	25818		&	26406\\
	\bf R						&	\bf 85				&	\bf 57		&	\bf 59		&	\bf 45		&	\bf 35\\
	N							&	841					&	1407		&	1458		&	1986		&	2934\\
	$\Delta_0$					&	29					&	19			&	18			&	13			&	9\\
	subscription				&	1.001				&	0.999		&	0.976		&	1.003		&	0.996\\
	Size of electrical group	&	435*				&	532*		&	486			&	520*		&	486*\\
	Number of groups			&	58					&	51			&	54			&	51			&	55\\
	Electrical cables			&	5684				&	620			&	10935		&	3381		&	25101\\
	Optical cables				&	17864				&	26094		&	18954		&	28395		&	13041\\
	\hline
	\bf Cost per node (\$)		&	\bf 1237.43			&	\bf 1314.29	&	\bf 1344.11	&	\bf 1497.77	&	\bf 1457.39\\
	\bf Power per node (W)		&	\bf 8.21			&	\bf 8.40	&	\bf 9.18	&	\bf 9.70	&	\bf 10.89\\
	\hline
	\end{tabular}
	\end{center}
	\caption{Example networks with about 25,000 compute nodes and electrical groups of about 500 nodes.}
	\label{tbl:25000}
\end{table*}

%% file: indirect.tex
\section{Indirect Networks}\label{sec:indirect}

Previous sections of this paper have studied direct networks, giving general bounds on the number of nodes for optimal topologies. Moreover, topologies that are close to these bounds have also been studied. However, indirect topologies are popular in the industry. For example, Clos networks have a widespread use since more or less half of current supercomputers on the Top 500 list are using them \cite{top500}. Hence, in this section it is explored how the cost model presented in this paper could be adapted to indirect networks. Moreover, the cost-optimal diameter 2 indirect network, which is the Two-Level Orthogonal Fat Tree~\cite{Valerio}, can also be obtained using the incidence graph of a projective plane. Hence, in this section it is also illustrated how the previous theoretical graph models for obtaining optimal direct networks can also be applied when dealing with indirect networks.

A \textsl{indirect network} has two types of routers since one router may or may not host compute nodes. Therefore, there are \textsl{spine routers}, which are connected only to other routers and \textsl{leaf routers}, which are also connected to compute nodes. Typically, all routers use the same hardware, so it can be assumed that every router has the same radix $R$. In addition, it will be assumed that all leaf routers have the same number $\Delta_0$ of attached compute nodes. Therefore, the graph defined by the routers has two kind of vertices: leaf vertices of degree $\Delta$ and spine vertices with degree $R$, which clearly implies that it cannot be vertex-transitive. Note that the relation $R=\Delta+\Delta_0$ considered for direct networks still holds in the case of indirect networks. In the following, the number of leaf routers will be denoted by $L$ and the number of spine routers by $S$. Thus, the total number of routers will be $N=L+S$.

When considering the graph model to study indirect networks, the main difference with the direct case lies on the diameter and average distance calculation.
In this case, the distances of interest are the ones between leafs, so that a great distance between some leaf and some spine routers becomes irrelevant.
Thus, instead of the diameter, the maximum distance among leafs is considered; and instead of average distance, the average distance between leafs, still denoted by $\bar k$.
In the remainder of the section it will be shown how the graph theoretical techniques presented in previous sections can be used to infer indirect network topologies with good properties.

A first example considers the indirect topology presented by Fujitsu in~\cite{Fujitsu}. This topology, denoted as Multi-layer Full-Mesh (MLFM), can be obtained from the incidence graph of a complete graph $K_n$. To explain this construction let us refer to Figure~\ref{fig:Fujitsu}. In this figure, the network is constructed using the incidence graph of $K_4$. In Figure~\ref{fig:Fujitsu} a) a standard representation of the incidence graph of $K_4$ is shown. The square shaped vertices are the vertices of the complete graph and the circle shaped are the vertices representing the incidence. For example, since there is a edge joining vertices $a$ and $b$ in $K_4$, vertex $A$ is adjacent to both in the incidence graph. In Figure~\ref{fig:Fujitsu} b) a different representation of this graph is shown, where vertices on the bottom are the vertices in $K_4$ and the vertices on the top are edges. Thus, the upper vertices will correspond to spine routers and the bottom vertices with leaf routers. Finally, in order to equalize the radix of the routers, leafs are replicated and compute nodes are added, as represented in Figure~\ref{fig:Fujitsu} c). In general, such a configuration can be obtained from any $K_n$, thus obtaining a indirect network topology with $\binom {n} {2}$ spine routers and $n(n-1)$ leaf routers, each one connected to $n-1$ compute nodes. Therefore, $\Delta = n-1$, $\Delta _0 = n-1$ and $R = 2 \Delta$. However, as it will be shown next, this topology is far from being the cost-optimal one among all the indirect topologies of diameter 2.

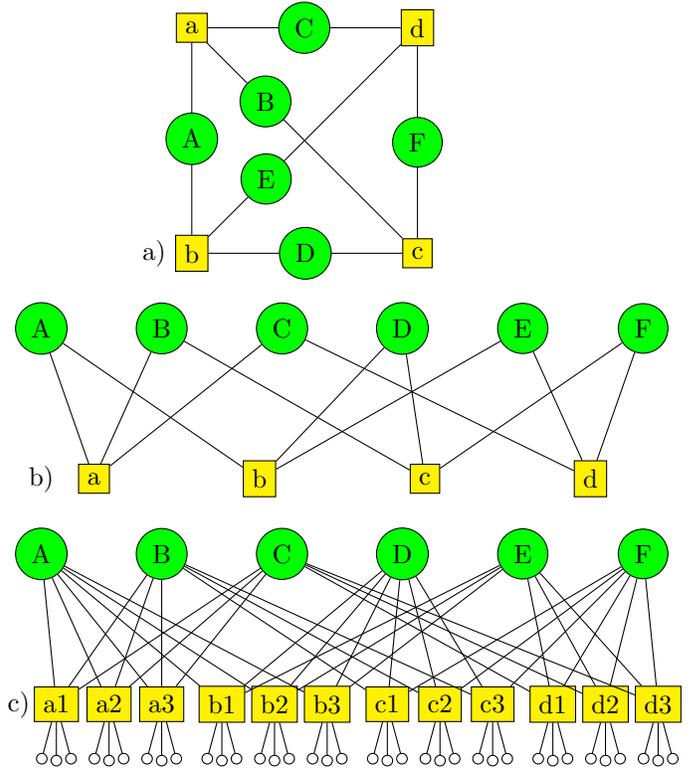
\begin{figure}
    \begin{center}
    \begin{tikzpicture}[
            up node/.style={draw,fill=green,circle},
            bottom node/.style={draw,fill=yellow},
        ]
    \begin{scope}[shift={(2,0)}]
        \node[bottom node] at (0,3) (vertexa) {a};
        \node[bottom node] at (0,0) (vertexb) {b};
        \node[bottom node] at (3,0) (vertexc) {c};
        \node[bottom node] at (3,3) (vertexd) {d};
        \draw (vertexa) --node[up node]{A} (vertexb);
        \draw (vertexa) --node[up node,pos=0.3]{B} (vertexc);
        \draw (vertexa) --node[up node]{C} (vertexd);
        \draw (vertexb) --node[up node]{D} (vertexc);
        \draw (vertexb) --node[up node,pos=0.3]{E} (vertexd);
        \draw (vertexc) --node[up node]{F} (vertexd);
        \node at (-.5,0) {a)};
    \end{scope}
    \begin{scope}[shift={(0,-3)}]
        \foreach \a/\av in {A/0,B/1,C/2,D/3,E/4,F/5}
        {
            \path (\av*1.6,2) node[up node] (up\a) {\a};
        }
        \foreach \b/\bv in {a/0,b/1,c/2,d/3}
        {
            \pgfmathsetmacro\x{.7+\bv*2.2}
            \path (\x,0) node[bottom node] (bottom\b) {\b};
        }
        \foreach  \a/\b in {A/a,A/b,B/a,B/c,C/a,C/d,D/b,D/c,E/b,E/d,F/c,F/d}
        {
            \draw (up\a) -- (bottom\b);
        }
        \node at (0,0) {b)};
    \end{scope}
    \begin{scope}[shift={(0,-6)}]
        \foreach \a/\av in {A/0,B/1,C/2,D/3,E/4,F/5}
        {
            \path (\av*1.6,2) node[up node] (up\a) {\a};
        }
        \foreach \b/\bv in {a/0,b/1,c/2,d/3}
        \foreach \i in {1,2,3}
        {
            \pgfmathsetmacro\x{\bv*2.2+(\i-1)*0.7+0.2}
            \path (\x,0) node[bottom node] (bottom\b\i) {\b\i};
            \foreach \angle in {-15,0,15}
            {
                \draw (bottom\b\i) -- +(\angle-90:.75) coordinate (x);
                \draw[fill=white] (x) circle (2pt);
            }
        }
        \foreach  \a/\b in {A/a,A/b,B/a,B/c,C/a,C/d,D/b,D/c,E/b,E/d,F/c,F/d}
        \foreach \i in {1,2,3}
        {
            \draw (up\a) -- (bottom\b\i);
        }
        \node at (-.3,0) {c)};
    \end{scope}
    \end{tikzpicture}
    \end{center}
    \caption{Incidence graph of $K_4$ and the Fujitsu network.}
    \label{fig:Fujitsu}
\end{figure}


An analysis for cost and power optimization as the one done in Section~\ref{sec:model} would be pleasing. Unfortunately, it is unfeasible due to, among other reasons, the hardness of calculating Moore bounds on irregular graphs.
Nevertheless, it is possible to infer a similar formula when it is assumed that the maximum distance between leaf routers is 2, as in the previous case of the Multi-layer Full-Mesh. For this purpose, let us consider that there might be links from a leaf router to another leaf router\footnote{links between spines are possible only for diameter $k \geq 3$.}. Therefore, let $\delta$ denote the number of links from a leaf router to another leaf router, which is again assumed to be constant.
Note that $\delta=\Delta$ in direct topologies and $\delta=0$ in fully indirect topologies, but there are some intermediate topologies. 
Now, since the maximum distance between leaf routers is 2, every of the $R$ links in a spine router must go to leaf routers. Thus, counting the links between leaf routers and spine routers it is obtained the following expression
	$$L(\Delta-\delta)=SR.$$

Now, the maximum number of leafs in a graph with maximum distance between leafs being 2, can be expressed in terms of $(\delta, \Delta, R)$ as follows:
\begin{equation}\label{eq:moore_indirect}
	L\leq 1+\delta^2+(\Delta-\delta)(R-1),
\end{equation}

Note that this is a Moore bound calculation but only considering leaf vertices. Also, if $\delta=\Delta$ then it becomes the original Moore bound $M(\Delta,2)$ presented in Equation~\eqref{eq:moore}.

The optimal value for the number of compute nodes is obtained when $$\Delta_0=\frac{u}{\bar k}(2\Delta-\delta),$$ which generalizes Equation~\eqref{eq:nodes}. Now, the cost per compute node is, analogously as it was done in Equation~\eqref{eq:cost},
$$\frac{\#\text{ports}}{\#\text{compute nodes}}=\frac{NR}{L\Delta_0}
	=\frac{R+\Delta-\delta}{\Delta_0}=1+\frac{\bar k}{u}.$$

This surprisingly implies that the cost per node does not depend on $\delta$. Hence, the most interesting value for $\delta$ would be the one giving the best scalability, since it provides the maximum number of compute nodes for the same cost. The maximum for Equation~\eqref{eq:moore_indirect} is obtained when $\delta=0$, which is the typical situation in indirect networks. That is, $$L\leq 1+\Delta(R-1).$$ There already exists a topology called \textsl{ Orthogonal Fat Tree} (OFT) presented in \cite{Valerio} that asymptotically attains this bound for $\bar k= 2$. This was already experimentally proved in \cite{Kathareios}. Next, a different construction than the one given in that work is presented, illustrating how also OFTs can be obtained from projective finite planes.

OFTs were constructed in \cite{Valerio} using orthogonal Latin squares. As the author already remarked in that paper, there is a intimate relation between orthogonal Latin squares and finite projective planes. That is, there are $n-1$ mutually orthogonal $n$-by-$n$ Latin squares if and only if there is a finite projective plane of order $n$~\cite{Colbourn}. Therefore, in the following definition, OFTs are built directly using projective spaces instead of manipulating mutually orthogonal Latin squares.

\begin{definition} Let $q$ be a power of a prime number. Let $\hat G_q = (V, E)$ be the graph with vertex set
$$V=\{ (s,P) \mid s\in\{0,1,2\},\ P\in P_2(\mathbb F_q) \}$$
and edge set
	$$E=\bigl\{ \{(0,P),(1,L)\},\{(1,P),(2,L)\} \mid P\perp L \bigr\}.$$
Thus, $\hat G_q$ is said to be the \textsl{orthogonal fat tree} of $P_2(\mathbb F_q)$.
\end{definition}

In a OFT network, vertices $(1,P)$ correspond to spine routers and the rest to leaf routers. As an example, let us consider Figure~\ref{fig:OFT}. In this figure black circles represent routers and white circles compute nodes.
As it can be seen, the routers are displayed into three columns of $q^2+q+1 = 7$ routers, since the total number of routers is $N=3(q^2+q+1) = 21$. The column in the middle would correspond to spine routers and the other two to leaf routers. It can also be seen that $\Delta=\Delta_0=q+1$ and $T=2(q+1)(q^2+q+1)$. Indirect networks are no longer vertex-transitive since there exist two different kind of vertices (spine and leaf). However, OFT is edge-transitive, so the utilization is exactly $u=1$. The average distance between leafs is exactly $\bar k=2$, since for any two leafs the minimal path connecting them is of length 2. Note that for each leaf there are several spine routers at distance 3. Finally, it is worthwhile to note that two $G_q$ projective networks are embedded in any $\hat{G}_q$, thus connecting these two different topologies. Moreover, it can be seen that this network has the same cost than the demi-PN and almost the same scalability of the PN, since $T_{PN} = 0.29 R^3$ and $T_{OFT}= 0.25 R^3$.

Finally, let us consider two different cases of use similar to the ones developed in subsection \ref{subsec:cases_of_use} but for indirect networks. Table \ref{table:OFT_case_of_use} presents the cost and power per node for OFT and MLFM networks with sizes about 10000 and 25000 computed nodes. A typical layout of indirect networks is done without electrical groups, which implies that every cable has been considered to be optical for the calculations. The MLFM results are similar to the demi-PN with slightly higher power. With respect to the OFT, on the one hand its scalability is slightly lower than PN, since with a slightly greater radix router it connects almost the same number of terminals. On the other hand, OFT has the same cost and power per node than the demi-PN.

\begin{table}
	\begin{center}\scriptsize
	\begin{tabular}{|c|cc|cc|}
	\hline
	Topology		&	MLFM 22			&	MLFM 30	&	OFT 16	&	OFT 23\\
	\hline
	T				&	9702			&	25230	&	9282	&	26544\\
	\bf R			&	\bf 42			&	\bf 58	&	\bf 34	&	\bf 48\\
	N				&	693				&	1305	&	819		&	1659\\
	$\Delta_0$		&	21				&	29		&	17		&	24\\
	cables			&	9702			&	25230	&	9282	&	26544\\
	\hline
	\bf Cost per node	&\bf 1297.18	&\bf 1321.76	& \bf 1282.19	&	\bf 1312.14\\
	\bf Watts per node	&\bf 8.4		&\bf 8.4		& \bf 8.4		&	\bf 8.4\\
	\hline
	\end{tabular}
	\end{center}
	\caption{Example Multi-Layer Full-Mesh and OFT networks with about 10,000 and 25,000 compute nodes.}
	\label{table:OFT_case_of_use}
\end{table}

\begin{figure}
	\begin{center}
	\begin{tikzpicture}
	\foreach \a in {0,1}
	\foreach \b in {0,1}
	\foreach \c in {0,1}
	{
		\pgfmathtruncatemacro\heigh{\a*4+\b*2+\c}
		\ifthenelse{\heigh=0}{}
		{
			\path[fill] (0,\heigh) node[anchor=south] {(0,(\a,\b,\c))} circle (2pt) coordinate (left \a \b \c)
				++(2,0) coordinate (center \a \b \c) circle (2pt) node [anchor=south] {(1,(\a,\b,\c))}
				++(2,0) coordinate (right \a \b \c) circle (2pt) node [anchor=south] {(2,(\a,\b,\c))}
				;
			\foreach \angle in {-10,0,10}
			{
				\draw (left \a \b \c) -- +(180+\angle:.75) circle (2pt);
				\draw (right \a \b \c) -- +(\angle:.75) circle (2pt);
			}
		}
	}
	\foreach \la in {0,1}
	\foreach \lb in {0,1}
	\foreach \lc in {0,1}
	\foreach \ra in {0,1}
	\foreach \rb in {0,1}
	\foreach \rc in {0,1}
	{
		\pgfmathtruncatemacro\goodl{\la==1 || \lb==1 || \lc==1}
		\pgfmathtruncatemacro\goodr{\ra==1 || \rb==1 || \rc==1}
		\pgfmathtruncatemacro\dot{mod(\la*\ra + \lb*\rb +\lc*\rc,2)}
		\pgfmathtruncatemacro\good{\goodl && \goodr && \dot==0}
		\ifthenelse{\good=1}
		{
			\draw (left \la \lb \lc) -- (center \ra \rb \rc) -- (right \la \lb \lc);
		}
	}
	\end{tikzpicture}
	\end{center}
	\caption{Orthogonal Fat Tree $\hat G _2$}
	\label{fig:OFT}
\end{figure}
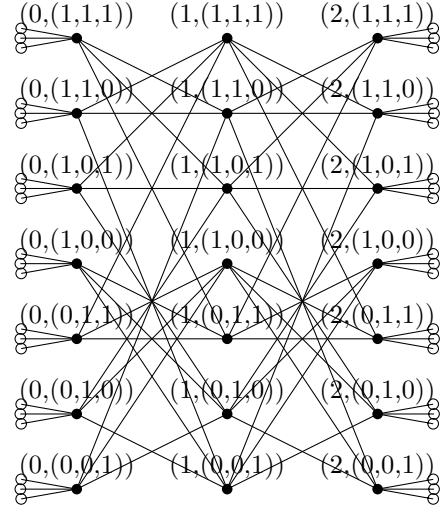

%% file: conclusions.tex
\section{Conclusions}\label{sec:conclusions}

Projective networks have been proposed in this paper for large systems using direct networks. These networks are built using incidence graphs of projective planes. Our proposal has been done by means of a coarse-grain cost model based on minimizing the average distance of the network while maintaining a uniform link utilization. The optimal networks under this cost model are those generalized Moore graphs which have uniform link utilization and, in particular, those being symmetric. By a complete a study of all the actually known families of generalized Moore graphs, for a given radix router and a number of compute nodes it is possible to choose the optimal network, using this cost model. In particular, projective networks have been proved to be a feasible alternative to the recently proposed Slim Fly. Finally, a first approach to the indirect networks' case has been considered. Our cost model has been adapted to this situation only for diameter two networks, since a general model for any diameter seems unfeasible. As it has been shown, optimal indirect networks for this case are the two-level Orthogonal Fat Trees, which can be also obtained by means of incidence graphs of projective planes.